\documentclass[lettersize,journal]{IEEEtran}
\usepackage{amsmath,amsfonts}

\usepackage{array}
\usepackage[caption=false,font=normalsize,labelfont=sf,textfont=sf]{subfig}
\usepackage{textcomp}
\usepackage{stfloats}
\usepackage{url}
\usepackage{verbatim}
\usepackage{graphicx}
\def\BibTeX{{\rm B\kern-.05em{\sc i\kern-.025em b}\kern-.08em
    T\kern-.1667em\lower.7ex\hbox{E}\kern-.125emX}}
\usepackage{balance}

\usepackage{mathtools}

\usepackage{acro}
\usepackage{siunitx}

\usepackage{amsthm}

\usepackage{makecell}

\usepackage{enumitem}

\def\fbinary{\mathbb{F}_2}
\newcommand{\freal}{\mathbb{R}}

\newcommand{\const}{\Gamma}
\newcommand{\seg}{\mathcal{S}}
\newcommand{\guessingIdx}{g}
\newcommand{\nc}{q}
\newcommand{\gaussnoise}{\xi}
\newcommand{\pnoise}{P_{\text{noise}}}
\newcommand{\plist}{P_\mathcal{L}}
\newcommand{\bima}[1]{\langle #1 \rangle}

\newcommand{\segi}{{(j)}}
\newcommand{\segis}[1]{{(#1)}}

\newcommand{\lEsNO}{E_\textnormal{s} / N_0}
\newcommand{\lEbNO}{E_\textnormal{b} / N_0}

\newtheorem{thm}{Theorem}

\usepackage{dsfont}
\newcommand{\ind}{\mathds{1}}

\usepackage[ruled,vlined, linesnumbered]{algorithm2e}

\SetKwFor{For}{for}{}{}
\SetKwFor{ForEach}{for each}{}{}
\SetKwFor{While}{while}{}{}
\SetKwIF{If}{ElseIf}{Else}{if}{}{else if}{else}{}

\DeclareAcronym{SNR}{
    short = SNR,
    long = signal-to-noise ratio,
}

\DeclareAcronym{GRAND}{
    short = GRAND,
    long = guessing random additive noise decoding,
}

\DeclareAcronym{ML}{
    short = ML,
    long = maximum likelihood,
}

\DeclareAcronym{ORBGRAND}{
    short = ORBGRAND,
    long = ordered reliability bits GRAND,
}

\DeclareAcronym{LLR}{
    short = LLR,
    long = log likelihood ratio,
}

\DeclareAcronym{TPC}{
    short = TPC,
    long = turbo product code,
}

\DeclareAcronym{LDPC}{
    short = LDPC,
    long = low-density parity-check
}

\DeclareAcronym{SISO}{
    short = SISO,
    long = soft-input soft-output
}

\DeclareAcronym{SOGRAND}{
    short = SOGRAND,
    long = soft-output GRAND,
}

\DeclareAcronym{SO}{
    short = SO,
    long = soft-output,
}

\DeclareAcronym{HW}{
    short = HW,
    long = Hamming weight,
}

\DeclareAcronym{LW}{
    short = LW,
    long = logistic weight,
}

\DeclareAcronym{BLER}{
    short = BLER,
    long = block error rate,
}

\DeclareAcronym{GR}{
    short = GR,
    long = guessing ratio
}

\DeclareAcronym{dRM}{
    short = dRM,
    long = dynamic Reed Muller
}

\DeclareAcronym{MS}{
    short = MS,
    long = min-sum
}

\DeclareAcronym{BP}{
    short = BP,
    long = belief propagation
}

\DeclareAcronym{GSegGRAND}{
    short = GSegGRAND,
    long = Generalized Segmented GRAND
}

\begin{document}
\title{Generalized Segmented GRAND for Guesswork Reduction in Turbo Product Decoding}
\author{Lukas Rapp,~\IEEEmembership{Graduate Student Member,~IEEE}, Jiewei~Feng, Muriel Médard,~\IEEEmembership{Fellow,~IEEE}, and Ken R. Duffy,~\IEEEmembership{Senior~Member,~IEEE}%
\thanks{This paper was presented in part at 2025 ISIT~\cite{rappBalancedTreeTransformation2025a}.}%
\thanks{L. Rapp and M. Médard are with the Research Laboratory for Electronics, Massachusetts Institute of Technology, Cambridge, MA 02139 USA (e-mail: rappl@mit.edu; medard@mit.edu).}%
\thanks{J. Feng and K. R. Duffy are with the Dept. of Mathematics and Dept. of ECE, Northeastern University, Boston, MA 02115 USA (e-mail: feng.ji@northeastern.edu; k.duffy@northeastern.edu).}%
\thanks{This work was supported by the Defense Advanced Research Projects Agency (DARPA) under Grant HR00112120008.}
}

\maketitle

\begin{abstract}
\Ac{GRAND} can efficiently decode any moderately redundant code with near \ac{ML} performance via noise effect guessing. 
For binary linear codes, Rowshan and Yuan's Segmented GRAND was the first to show that constrained guessing can reduce guesswork. Although powerful, their approach requires a specific parity-check matrix structure that limits the number of constraints that can be exploited as well as the class of applicable codes. Here we introduce GSegGRAND, a generalization of Segmented GRAND that circumvents its limitations. Built on a novel parity check structure and a transformation that maps codes into this structure, 
GSegGRAND efficiently incorporates up to \(\log_2(n)\) constraints for a wide range of codes, reducing guesswork by an additional \(75 \%\) over Segmented GRAND. To leverage that advantage for soft-output decoding, we derive an accurate \ac{SO} equation for GSegGRAND by extending \ac{SOGRAND} to incorporate constrained guessing. Applying this \ac{SO} to turbo product decoding, GSegGRAND achieves up to \(88 \%\) guesswork reduction, making it a promising candidate for low-latency decoding in future communication systems.
\end{abstract}


\acresetall

\section{Introduction}
\IEEEPARstart{G}{uessing} random additive noise decoding (GRAND)~\cite{duffyCapacityAchievingGuessingRandom2019} \acuse{GRAND}
is a decoding framework capable of decoding any code of moderate redundancy, including non-linear codes~\cite{cohenAESErrorCorrection2023}, while achieving near \ac{ML} performance. A core component is a noise effect generator that produces noise effects in decreasing order of likelihood. For each query, \ac{GRAND} inverts the noise effect from the hard-decision channel output and returns the resulting vector as soon as it is a valid codeword. This principle can be applied to various channel models, such as hard-decision~\cite{duffyCapacityAchievingGuessingRandom2019}, soft-decision~\cite{duffyOrderedReliabilityBits2022}, or correlated channels~\cite{anKeepBurstsDitch2022, duffyUsingChannelCorrelation2023a}. \ac{GRAND}'s practicality has been demonstrated through multiple taped out chips~\cite{riazMultiCodeMultiRateUniversal2021a, blancGRANDABDecoder8482024,riazSub08pJBitUniversal2025,kizilatesLowLatencyModulationCorrelationAdaptive2025,Kizilates26SOGRAND} and hardware syntheses, e.g. ~\cite{condoHighperformanceLowcomplexityError2021, condoFixedLatencyORBGRAND2022, abbasHighThroughputEnergyEfficientVLSI2022, chu2023efficient,ji2025efficient,abbas2025improved}.

Recently, there has been increased interest in the reduction of \ac{GRAND}'s \emph{guesswork}, i.e., the number of noise effect queries until a codeword is found or decoding is abandoned. Existing work either targets special codes and applications~\cite{liSCLGRANDLowerComplexity2023a, zhouHybridGRANDSphere2023} or is applicable to general linear codes~\cite{rowshanConstrainedErrorPattern2022a,rowshanLowComplexityGRANDSegmentation2023,chuEfficientHardDetectionGRAND2023, chatzigeorgiouSymbolLevelGRANDHighOrder2023, rezaOrdeptJournal2025,wangPartiallyConstrainedGRAND2024, abbasImprovedStepGRANDLowLatency2025, 
rowshanJournalNew2025, rappSOGRANDAssistedGuesswork2025}.
One strategy is to exploit the syndrome of each query during decoding, enabling \cite{rezaOrdeptJournal2025} to correct an additional bit error and \cite{chuEfficientHardDetectionGRAND2023} to terminate decoding early. Another promising direction is \emph{constrained guessing}.
For a binary linear code with parity check matrix \(H\), \ac{GRAND} searches for a noise effect \(z^n\) that satisfies a set of parity constraints \(H z^n = s^{n-k}\), induced by the syndrome \(s^{n-k}\) of the received hard-decision. Guesswork can therefore be reduced by restricting the guessing to noise effects that already fulfill a subset of such constraints, where each constraint reduces guesswork by up to a factor of \(2\)~\cite{rowshanJournalNew2025}. This creates the challenge of designing noise effect generators that can incorporate constraints efficiently while preserving an accurate guessing order.

Rowshan and Yuan's proposed Segmented GRAND~\cite{rowshanConstrainedErrorPattern2022a,rowshanLowComplexityGRANDSegmentation2023,rowshanJournalNew2025}, which generates constrained queries efficiently using multiple noise effect generators, each operating on distinct segments of the codeword bits. That approach requires that the parity check matrix be in a class that can be transformed into a special structure where there exists a subset of rows such that each row has 1s in distinct, non-overlapping columns. If the parity check matrix can be transformed into that form, for each of these rows, the columns containing 1s then define an individual segment. In this case, segment parities correspond to syndrome values, allowing each noise effect generator to skip the generation and testing of noise effects of incorrect parity. While Segmented GRAND reduces guesswork efficiently, its guesswork reduction is limited by the number of rows that can be transformed into the specific structure.
The transformation becomes more challenging as more constraints are extracted, and an efficient method for extracting numerous constraints for any given code is currently lacking.

In this paper, we propose \ac{GSegGRAND}, which operates on a new parity check structure into which a wide range of parity check matrices can be easily transformed.
This structure and transformation, which we term \emph{balanced tree structure} and \emph{balanced tree transformation}, preserves the codebook and facilitates the extraction of additional constraints.
For all practical codes that we tested, we consistently extracted \(4\) constraints enabling a guesswork reduction by a factor of \(16\) over \ac{ORBGRAND} without constrained guessing. To achieve this, we propose a novel noise effect generator tailored to this generalized structure. Unlike Segmented GRAND, whose segment parities are fixed, \ac{GSegGRAND} works with segment parities that form a linear binary subspace induced by the balanced tree structure. This paper demonstrates how \ac{GSegGRAND} can efficiently iterate through noise effects within this subspace.

An important recent development of \ac{GRAND} is \ac{SOGRAND}~\cite{yuanSoftoutputGRANDLong2023}, which produces accurate \ac{SO} and facilitates iterative decoding of \ac{TPC}. \Acp{TPC} enable the construction of long and powerful codes that, when decoded with \ac{SOGRAND}, can outperform \ac{LDPC} codes used in the 5G New Radio~\cite{yuanSoftoutputGRANDLong2023}. 

\ac{GRAND} variants with constrained guessing, such as Segmented GRAND or \ac{GSegGRAND}, have the potential to reduce guesswork during \ac{TPC} decoding, an application that remains largely unexplored. However, constrained guessing requires a more nuanced \ac{SO} calculation: Feng et al. found that, when using \ac{ORBGRAND} with noise effect skipping for even-weight codes, conditioning the noise effect probabilities on their parity improves the \ac{SO} quality~\cite{fengLeveragingCodeStructure2025}. As \ac{GSegGRAND} exploits multiple constraints for skipping, a generalization of that correction is necessary.

In this paper, we propose a general framework to improve \ac{SOGRAND}'s \ac{SO} for any \ac{GRAND} algorithm with constrained guessing and demonstrate it on \ac{GSegGRAND}. Our simulation results highlight the importance of the improved \ac{SO} calculation for turbo product decoding: with the conventional \ac{SO} calculation, constrained guessing suffers from a high error floor above a \ac{BLER} of \num{e-1} for the \((32, 26)^2\) eBCH \ac{TPC}. In contrast, with the proposed \ac{SO} calculation, \ac{GSegGRAND} reaches a \ac{BLER} performance close to \ac{SOGRAND} while preserving most of the guesswork reduction observed during isolated component code decoding.

This paper is structured as follows: Section~\ref{sec:preliminaries} introduces the preliminaries such as \ac{GRAND}, Segmented \ac{GRAND} and turbo product decoding. 
Section~\ref{sec:genseggrand} outlines \ac{GSegGRAND} and its balanced tree transformation while Sec.~\ref{sec:socalculation} develops the improved \ac{SO} calculation for constrained guessing. Section~\ref{sec:results} presents the simulation results of \ac{GSegGRAND} under component decoding and turbo product decoding.

\begin{table}[t]
\caption{Summary of Important Notation}
\label{tab:notation_summary}
\centering
\setlength{\tabcolsep}{3pt}
\begin{tabular}{l l}
\hline
\textbf{Symbol} & \textbf{Description \quad(Defined in)} \\
\hline
\(c^n\) & Transmitted codeword (\ref{sec:channel}). \\
\(y^n, y_\text{hd}^n, \ell^n\) & Received sequence, hard-decision and LLRs  (\ref{sec:channel}). \\
\(s^{n-k}\) & Syndrome of the received hard-decision (\ref{sec:preliminaries-grand}). \\
\(e^n\) & True noise effect (\ref{sec:channel}).\\
\(z^{n,(i)}, z^n\) & \(i\)-th, and generic noise effect query (\ref{sec:preliminaries-grand}). \\
\(\nc, t\) & Number of constraints and segments (\ref{sec:seg-grand}).\\
\(n_\segi\) & Length of segment \(j\) (\ref{sec:seg-grand}).\\
\(z_\segi^{n_\segi}, e_\segi^{n_\segi}\) & Subeffect of \(z^n, e^n\) of segment \(j\) (\ref{sec:seg-grand}).\\
\(\ell_\segi^{n_\segi}\) & LLRs of \(\ell^n\) of segment \(j\) (\ref{sec:dec-alg-seg-grand}).\\
\(w_{\text{H},\segi}, w_{\text{L},\segi}, p_\segi\) & HW, LW, and parity of subeffect \(z_\segi^{n_\segi}\) (\ref{sec:seg-grand}).\\
\(w_\text{H}, w_\text{L}, w_\text{T}\) & HW, LW, and Total LW of \(z^n\) (\ref{sec:orbgrand}, \ref{sec:seg-grand}).\\
\(\bima{v_1 \dots v_\nc}\) & Integer value \(\sum_{i=1}^\nc v_i 2^{\nc-i}\) of \(v^\nc \in \fbinary^n\) (\ref{sec:bal-tree-struc})\\
\(p_\text{u}^\nc, p_\text{nu}^{t-\nc-1}\) & Unit and non-unit parities (\ref{sec:bal-tree-struc}). \\
\(\hat{c}^n\), \(\mathcal{L}\) & Codeword candidate and candidate list (\ref{sec:preliminaries-grand}, \ref{sec:prelim-sogrand}).\\
\(\pnoise, \plist\) & Accumulated noise effect and list probability (\ref{sec:prelim-sogrand}). \\
\hline
\end{tabular}
\end{table}

\section{Preliminaries}\label{sec:preliminaries}
\subsubsection{Notation}
In this work, sets are denoted by calligraphic uppercase letters (e.g., \(\mathcal{A}\)), and the cardinality of a set \(\mathcal{A}\) is denoted by \(|\mathcal{A}|\). We use $[n]$ to represent the set of integers $\{1,\ldots,n\}$. Random variables are denoted by uppercase letters (e.g., \(A\)), while specific realizations are denoted with lowercase letters (e.g., \(a\)).  A sequence or vector of length \(n\) is represented with a superscript, \(a^n\), where \(a_i\) denotes the \(i\)-th entry. Unless otherwise specified, vectors are indexed starting at \(1\) (\(i \in [n]\)); the only exception is for bit segments whose indexing begins at \(0\) for convenience. Matrices are likewise denoted by uppercase letters \(A\); the distinction between matrix and a random variable will be clear from context.
We adopt MATLAB-style slice notation. For a matrix \(A\) and integers \(a, b\), \(A_{a:b, :}\) denotes the submatrix formed by rows \(a\) through \(b\) (inclusive) and all columns, while \(A_{:, a:b}\) denotes all rows and columns \(a\) through \(b\). We use \(\fbinary = \{0, 1\}\) to denote the binary field. For two binary sequences \(a^n, b^n\), we denote the componentwise addition in \(\fbinary\) by \(a^n \oplus b^n\), and the parity via \(\Phi(a^n) = \sum_{i=1}^n a_i \mod 2\). Table~\ref{tab:notation_summary} lists the important notation used in this work. The section or equation in which the symbol is defined is listed at the end of each description.

\subsubsection{Channel}\label{sec:channel}
While the \ac{GRAND} principle works with any code, including non-linear codes~\cite{cohenAESErrorCorrection2023}, and various channel models, we consider in this work linear codes, and assume that the channel outputs \acp{LLR}. Specifically, let \(\mathcal{C} \subset \fbinary^n\) and \(H \in \fbinary^{m \times n}\) be the codebook and parity check matrix of an \((n,k)\) linear code with information and codeword length \(k\) and \(n\) and \(m \coloneqq n - k\) parity check constraints. Let \(c^n \in \mathcal{C}\) be a codeword that is transmitted over a binary symmetric memoryless channel, and \(y^n\) the sequence that is received. The receiver computes \acp{LLR} \(\ell^n \in \freal^n\) and hard-decision \(y_\text{hd}^n\) from \(y^n\), where \(y_{\text{hd},i} = 0\) if \(\ell_i > 0\) and \(y_{\text{hd},i} = 1\), otherwise. \(e^n = y_{\text{hd}}^n \oplus c^n\) denotes the true noise effect by which the channel perturbs \(c^n\), whose bit and block-wise a posteriori probability are given by
\begin{equation}\label{eqn:reliabilities}
    \begin{aligned}
        P_{E_i | Y_i}(e_i | y_i) &= (1 + \exp(-(-1)^{e_i} \ell_i))^{-1} \\
        P_{E^n | Y^n}(e^n | y^n) &= \prod_{i=1}^n P_{E_i | Y_i}(e_i | y_i)
        = \exp(-\text{Rel}(e^n)), \; \\
        \text{where} \; \text{Rel}(e^n) 
        &\coloneqq \sum_{i=1}^n |\ell_i| e_{i}.
    \end{aligned}
\end{equation}

\subsubsection{GRAND}\label{sec:preliminaries-grand}
The \ac{GRAND} principle~\cite{duffyCapacityAchievingGuessingRandom2019} enables universal decoding of any code of moderate redundancy. The key component of every \ac{GRAND} algorithm is a noise effect generator \(\textnormal{NoiseGen}(y^n) = [z^{n,(1)}, z^{n,(2)}, \dots]\) that iterates through noise effects in approximately decreasing order of a posteriori probability \(P_{E^n | Y^n}(z^{n,(i)} | y^n)\), where $z^{n,(i)}$ denotes the $i$-th generated binary noise effect. This implies that, for every $i$,
\(
    P_{E^n | Y^n}(z^{n, (i)} | y^n) \geq P_{E^n | Y^n}(z^{n, (i+1)} | y^n).
\)
For each query \(z^{n,(i)}\), \ac{GRAND} subtracts \(z^{n,(i)}\) from the hard decision and returns the first sequence that forms a valid codeword: \(y_\text{hd}^n \oplus z^{n,(i)} \in \mathcal{C}\), equivalent to 
\begin{equation}\label{eqn:parity-check-equations}
    H (y_\text{hd}^n \oplus z^{n,(i)}) = 0^{n-k} \Leftrightarrow
    H z^{n,(i)} = s^{n-k},
\end{equation}
where \(s^{n-k} \coloneqq H y_\text{hd}^n\) is the received syndrome. For clarity, we henceforth drop the superscript index \(i\) in \(z^{n,(i)}\): \(z^n\) and restore it only when the ordering of queries matters. 
If list decoding is desired, GRAND algorithms can continue the search until the required number of valid codewords is identified, yielding a list of candidate codewords \(\mathcal{L}\).
Decoding is abandoned if the number of guesses exceeds a given abandonment threshold~\cite{duffyCapacityAchievingGuessingRandom2019}. 

\subsubsection{ORBGRAND}\label{sec:orbgrand}
One soft-input \ac{GRAND} variant is \ac{ORBGRAND}~\cite{duffyOrderedReliabilityBits2022}, which uses an approximate guessing order to enable efficient hardware implementations. For clarity, this section assumes that the received \acp{LLR} \(\ell^n\) have increasing reliability \(|\ell_i|\). In practice, the received values are sorted upon reception.\footnote{Note that a full sorting is not necessary as \ac{ORBGRAND} can start decoding once the least reliable bits are identified, and incrementally sorts additional values until a valid codeword is found~\cite{riazSub08pJBitUniversal2025}.}
The noise effect generator of \ac{ORBGRAND} approximates the rank-ordered reliabilities \(|\ell_i|\) by a line
\(
    |\ell_i| \approx \beta i
\)
for \(i \in [n]\)
with slope \(\beta > 0\). With this approximation, the reliability of query \(z^n\) in \eqref{eqn:reliabilities} is proportional to the \ac{LW} of \(z^n\)~\cite{duffyOrderedReliabilityBits2022}
\begin{align}
    w_\textnormal{L} \coloneqq -\text{Rel}(z^n) / \beta \approx \sum_{i = 1}^n i z_i,\label{eq_ORB_rel}
\end{align}
implying that the noise effect probability decreases with \(w_\textnormal{L}\).
Consequently, \ac{ORBGRAND}'s noise effect generator approximates the optimal guessing order by iterating through noise effects in increasing \ac{LW}. 
To realize this, \ac{ORBGRAND} iterates for a given \ac{LW} \(w_\text{L}\) over all \acp{HW} \(w_\text{H}\) of \(z^n\) for which this \ac{LW} can be achieved, where \(w_\text{H} = \sum_{i = 1}^n z_i\).
For each \((w_\text{L}, w_\text{H})\) pair, the Landslide algorithm~\cite[Alg.~2]{duffyOrderedReliabilityBits2022} \(\text{LS}(w_\text{L}, w_\text{H})\) efficiently iterates through all noise effects of such given \ac{LW} and \ac{HW}. In this paper, Landslide's ability to generate noise effects of a specific \ac{HW} is leveraged to only generate noise effects that satisfy a set of binary constraints.

\subsubsection{SOGRAND}\label{sec:prelim-sogrand}
Recently, \ac{SOGRAND}~\cite{yuanSoftoutputGRANDLong2023}, a \ac{SISO} decoder, was proposed, which calculates block- and bit-wise \ac{SO} during the execution of any \ac{GRAND} algorithms including \ac{ORBGRAND}.
During decoding with $|\mathcal{L}|\geq1$, \ac{SOGRAND} accumulates the probability of the guessed noise effects and candidate codewords \(\hat{c}^n \in \mathcal{L}\)
\begin{equation}\label{eqn:prob-acc}
\begin{aligned}
    \pnoise &= \sum_{i=1}^\guessingIdx 
    P_{E^n | Y^n}\left(z^{n,(i)} | y^n\right), \\
    \plist &= \sum_{\mathclap{\hat{c}^n \in \mathcal{L}}}
    P_{E^{n} | Y^n}\left(\hat{c}^n \oplus y_\text{hd}^n | y^n\right),
\end{aligned}
\end{equation}
where \(\guessingIdx\) denotes the number of generated noise effects. 
The probability that each candidate codeword \(\hat{c}^n \in \mathcal{L}\) is correct and the probability that the correct codeword \(c^n\) is not in \(\mathcal{L}\) can then be estimated as~\cite{yuanSoftoutputGRANDLong2023}
\begin{equation}
\label{eqn:blockso}
\begin{aligned}
    \hat{P}_{C^n | Y^n}(\hat{c}^n | y^n) 
    &=
    \frac{
        P_{E^n | Y^n}\left(\hat{c}^n \oplus y_\text{hd}^n | y^n\right)
    }{
        \plist  + 
        (1 - \pnoise) 2^{-(n-k)}
    }, \\
    \hat{P}_{C^n | Y^n}(\mathcal{C} \setminus \mathcal{L} | y^n) 
    &=
    \frac{
        (1 - \pnoise) 2^{-(n-k)}
    }{
        \plist  + 
        (1 - \pnoise) 2^{-(n-k)}
    },
\end{aligned}
\end{equation}
respectively. 
Eq. \eqref{eqn:blockso} normalizes each codeword candidate's posterior probability by approximating the total codebook mass as \(\sum_{\hat{c}^n \in \mathcal{C}} P(\hat{c}^n | y^n) \approx \plist + (1 - \pnoise) 2^{-(n-k)}\), assuming a random codebook whose unexplored codewords are uniformly distributed among the unexplored noise effects.
Numerical studies confirm that this approximation also yields accurate soft-output for structured codes.

Besides this block-wise \ac{SO}, bit-wise \ac{SO} can be obtained by marginalizing the block-wise \ac{SO} as follows
\begin{equation}\label{eqn:bitso}
    \begin{aligned}
        &\hat{P}_{C_i | Y^n}(b | y^n) =
        \sum_{\hat{c}^n \in \mathcal{L}: \hat{c}_i = b}
        \hat{P}_{C^n | Y^n}(\hat{c}^n | y^n) \\
        &+ \hat{P}_{C^n | Y^n}(\mathcal{C} \setminus \mathcal{L} | y^n)
        P_{E_i | Y_i}(b \oplus y_{\text{hd},i} | y_i),
    \end{aligned}
\end{equation}
for \(b \in \{0, 1\}\) and \(i \in [n]\). These marginals can be compressed in a posterori \acp{LLR}
\begin{equation}\label{eqn:llrso}
    \ell_{\text{APP},i} = \ln\left(\hat{P}_{C_i | Y^n}(0 | y^n) / \hat{P}_{C_i | Y^n}(1 | y^n)\right).
\end{equation}
List decoding terminates if either a predefined maximal list size \(N_\mathcal{L}\) is reached or if the estimated list error probability is below a threshold \(T \in [0, 1]\).

\subsubsection{Turbo Product Decoding}
One application of \ac{SOGRAND} is turbo product decoding~\cite{pyndiahNearoptimumDecodingProduct1998}.
\Acp{TPC} are an efficient way to construct powerful long codes by concatenating short component codes, which can be efficiently decoded via turbo decoding. In this paper, we consider product codes based on an \((n, k)\) component code \(\mathcal{C}_\text{c}\). The resulting product code has rate \((k/n)^2\), where each codeword is represented by a binary $n \times n$ matrix $c^{n\times n}$~\cite{eliasErrorfreeCoding1954a}. A matrix is a valid product codeword if every row and column forms a valid codeword of \(\mathcal{C}_\text{c}\).

Let $c_{i,j}$ denote the entry of $c^{n\times n}$ at the $i$-th row and $j$-th column.
Let \(\ell_{\text{Ch},i,j}, \ell_{\text{A},i,j}, \ell_{\text{APP}, i,j}, \ell_{\text{E}, i,j}\) denote the channel, a priori, a posterori and extrinsic \ac{LLR} corresponding to bit \(c_{i,j}\), where \(i, j \in [n]\). For decoding, each bit is initially assigned an a priori \ac{LLR} \(\ell_{\text{A},i,j} = 0\). These \acp{LLR} are updated through multiple decoding iterations. In each iteration, all rows are first updated, followed by an update of all columns.
To update the \(i\)-th row, the \ac{LLR} vector \(\ell_{\text{A},i,:} + \ell_{\text{Ch},i,:}\) is input into \ac{SOGRAND}, yielding the a posteriori \ac{LLR} vector \(\ell_{\text{APP}, i,:}\) via \eqref{eqn:llrso}. Then, extrinsic \acp{LLR} \(\ell_{\text{E}, i,j} = \ell_{\text{APP}, i,j} - \ell_{\text{Ch}, i,j} - \ell_{\text{A},i,j}\) are extracted to update the a priori \acp{LLR} \(\ell_{\text{A},i,j} \gets \alpha \ell_{\text{E}, i,j}\) for \(j \in [n]\), where \(\alpha \in [0, 1]\) is a dampening factor~\cite{pyndiahNearoptimumDecodingProduct1998}. The column update follows the same procedure but slices along columns instead of rows. Iterative decoding terminates when either the current a posteri hard-decision \(\ind_{\{\ell_{\text{APP}, i,j} < 0\}}\) for \(i, j \in [n]\) forms a valid product codeword, or a maximum number \(I_\text{max}\) of full iterations is reached, in which case the last hard-decision is output. Here, \(\ind_{\{\ell_{\text{APP}, i,j} < 0\}}\) denotes the indicator function, which is \(1\) if \(\ell_{\text{APP}, i,j} < 0\) and \(0\) otherwise.  We refer to a row or column update as a half-iteration, and to a consecutive row and column update as a full-iteration.

\subsubsection{Segmentation and Segmented GRAND}\label{sec:seg-grand}
Recently, Rowshan and Yuan's proposed Segmented \ac{GRAND} and demonstrated that the average number of guesses of \ac{ORBGRAND} can be reduced if the parity check matrix fulfills specific properties~\cite{rowshanJournalNew2025}.
We outline the core idea of \cite{rowshanJournalNew2025} in this section and introduce notation that we later use to generalize the algorithm.
The core idea of Segmented \ac{GRAND} is to restrict noise effects generation to those effects that already fulfill a subset of the parity check equations in \eqref{eqn:parity-check-equations}. Each additional constraint can reduce the average guesswork by up to a factor of \(2\). To achieve this, Segmented \ac{GRAND} requires that the parity check matrix \(H\) can be transformed into a specific structure (see Fig.~\ref{fig:structure-segmented-grand}) via elementary row operations and column permutations.
\begin{figure}
    \centering
    \includegraphics{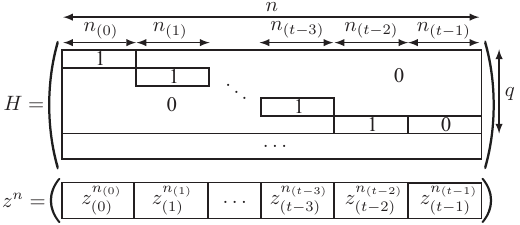}
    \caption{Structure of the parity check matrix required for Segmented GRAND}
    \label{fig:structure-segmented-grand}
\end{figure}
To restrict the noise effect generation by \(\nc\) constraints, the parity check matrix must contain \(\nc\) rows, of which at most one row has a \(1\) at each bit position. We now discuss such a requirement in more detail. 
In each of the first $\nc$ rows of $H$, the columns that contain a \(1\) in a specific row and the columns that contain a \(0\) in all \(\nc\) rows are partitioned into \(t = \nc + 1\) \emph{segments}, which are index sets
\begin{align*}
    \seg_\segi &= \{i: i \in [n], H_{j+1, i} = 1\}, \; \text{for \(j \in \{0, \dots, t-2\}\),}\\
    \seg_\segis{t-1} &= \{i: i \in [n], H_{j^\prime, i} = 0 \; \forall j^\prime \in [\nc]\}.
\end{align*}
To simplify notation later for the proposed generalization, the indexing of segments begins at \(0\). Segmented GRAND requires that the parity check matrix can be transformed such that \(\seg_{\segi} \cap \seg_{\segis{j^\prime}} = \emptyset\) for all $j\neq j^\prime$, as visualized in Figure 1 after column permutation. Throughout the paper, we permute the columns of the parity check matrix such that segment bits are grouped for visualization. 

We denote segment-specific properties via a subscript in brackets, e.g., \(\segi\): \(n_\segi \coloneqq |\seg_\segi|\) denotes the length of the \(j\)-th segment. \(e_\segi^{n_\segi} = (e_i)_{i \in \seg_\segi}\) and \(z_\segi^{n_\segi} = (z_i)_{ i \in \seg_\segi}\) are subsequences of the true noise effect \(e^n\) and current noise effect query \(z^n\), respectively, containing the bit values at the positions of segment \(j\), ordered by their original index $j$. We refer to these subsequences of noise effects as \emph{subeffects}. Let $z_{\segi,i}$ denote the $j$-th entry of $z^{n_\segi}_\segi$.
The quantities
\[
    w_{\text{H}, \segi} = \sum_{i = 1}^{n_\segi} z_{\segi,i}, 
    \quad w_{\text{L},\segi} = \sum_{i=1}^{n_\segi} i z_{\segi,i}, 
    \quad p_\segi = \Phi(z^{n_\segi}_\segi),
\]
denote the \ac{HW}, \ac{LW}, and parity of the \(j\)-th subeffect \(z^{n_\segi}_\segi\) of the current guess \(z^n\) and \( w_{\text{H}}^t, w_\text{L}^t\) and \(p^t = (p_\segi)_{j=0}^{t-1}\) denote the corresponding vectors.
We refer to the sum of the subeffects logistic weights as total \ac{LW} \(w_\text{T} = \sum_{j=0}^{t-1} w_{\text{L},\segi}\).
Since the segments are a partition of \([n]\), the subeffects form a segmentation of \(z^n\) and \(e^n\), respectively. For instance, for \(z^n\):
\(
    z^n = (z^{n_\segis{0}}_\segis{0} | z^{n_\segis{1}}_\segis{1} | \dots | z^{n_\segis{t-1}}_\segis{t-1}).
\)

Rowshan and Yuan observed that for this particular structure, the first \(\nc\) equations of \eqref{eqn:parity-check-equations} are fulfilled if and only if the parity of the \(j\)-th subeffect matches the \((j+1)\)-th syndrome, i.e., \(p_\segi = s_{j+1}\) for \(j \in \{0, \dots, \nc-1\}\). Segmented GRAND leverages this observation as follows: the outer iteration iterates over the total \ac{LW} \(w_\text{T}\), similarly to ORBGRAND, followed by a two-level partitioning that generates the noise effects. In the outer partition, all integer partitions of a total \ac{LW} \(w_\text{T}\) into \(t\) segments are generated: \(w_{\text{L}, \segis{0}}, \dots, w_{\text{L}, \segis{t-1}}\). For each segment \(j\), the inner partition generates all distinct integer partitions that sum up to \(w_{\text{L}, \segi}\), which correspond to the bit flip positions. By merging these segment-wise integer partitions, noise effects are constructed that, by design, fulfill the first \(\nc\) equations of \eqref{eqn:parity-check-equations} and have a total \ac{LW} of \(w_\text{T}\).

\section{Generalized Segmented GRAND}\label{sec:genseggrand}
While Segmented GRAND demonstrates that segmentation can enable guesswork reduction, it can only constrain the guessing for parity check rows that are in a specific structure.
In this section, we propose a generalized version of Segmented GRAND, termed \ac{GSegGRAND}, that enables the extraction of multiple constraints for a wide range of linear codes. The basis of this algorithm is the \emph{balanced tree structure}.
We will first introduce the tree structure and the balanced tree structure in Sec.~\ref{sec:bal-tree-struc}, followed by the corresponding decoding algorithm in Sec~\ref{sec:dec-alg-seg-grand}. In addition, a systematic procedure to transform most practical codes into a balanced tree structure will be provided in Sec.~\ref{sec:bal-tree-transf}. Finally, additional implementation details of the decoding algorithm are provided in Sec.~\ref{sec:implementation}.

\subsection{Balanced Tree Structure}\label{sec:bal-tree-struc}
\begin{figure}
    \centering
    \includegraphics{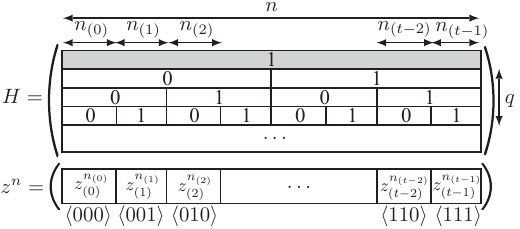}
    \caption{Balanced tree structure for \ac{GSegGRAND}. For even-weight codes, the structure can be extended by an additional all-one row (gray).}
    \label{fig:tree-structure}
\end{figure}
For a binary string \(v^\nc = (v_1,\dots,v_\nc) \in \fbinary^\nc\), we define its integer value as
\(
    \bima{v_1 v_2 \dots v_\nc} \coloneqq \sum_{i=1}^\nc v_i 2^{\nc-i}.
\)
We say that the first \(\nc\) rows of a parity-check matrix \(H\) are in tree structure if the columns of \(H\) are arranged in ascending order of their first \(\nc\) entries~\cite{rappErrorandErasureDecodingProduct2022} (see Fig.~\ref{fig:tree-structure}), i.e.,
\[  
    i < i^\prime \Rightarrow \bima{H_{1, i} H_{2, i} \dots H_{\nc,i}} \leq \bima{H_{1, i^\prime} H_{2, i^\prime} \dots H_{\nc,i^\prime}}.
\]
Any given \(H\) can be transformed into a tree structure by permuting its columns so that columns with identical first \(\nc\) entries are grouped into a segment \(j\), where \(j\) is the integer represented by those \(\nc\) binary entries. More explicitly, two columns $i,i^\prime$ are put in the same segment $j$ if $j= \bima{H_{1,i} H_{2,i} \dots H_{\nc,i}} = \bima{H_{1,i^\prime} H_{2,i^\prime} \dots H_{\nc,i^\prime}}$. This yields at most \(t = 2^\nc\) segments. To leverage this structure, we index segment \(j = \bima{v_1 \dots v_\nc} \in \{0, \dots, t-1\}\) interchangeably by its decimal \(j\) or binary representation \((v_1,\dots,v_\nc)\). Hence, for a segment parity we may write either \(p_\segi\) or \(p_\segis{\bima{v_1 v_2 \dots v_\nc}}\). Analogous to Segmented GRAND, the index set of segment \(\bima{v_1 \dots v_\nc}\) is given by
\[
    \seg_\segis{\bima{v_1 \dots v_\nc}} = \{i \in [n]: H_{1:\nc, i} = (v_1, \dots, v_\nc)\}.
\] 
As defined in Sec.~\ref{sec:seg-grand}, \(n_\segi, z_\segi^{n_\segi}, w_{\text{L},\segi}, w_{\text{H},\segi}\) and \(p_\segi\) denote length, subeffect, \ac{LW}, \ac{HW} and parity of segment \(j\). The tree structure is \emph{balanced} if all \(t = 2^\nc\) segments exist and the segments' lengths are approximately equal: \(n_\segis{\bima{v_1 \dots v_\nc}} \approx n / t\). 
In the following, we assume that the code's parity check matrix is in balanced tree structure. While this may appear restrictive, Sec.~\ref{sec:bal-tree-transf} demonstrates that a wide range of linear codes can be rewritten into this form.

By segmenting \(z^n\) into subeffects, the first \(\nc\) rows of \(s^{n-k} = H z^n\) in \eqref{eqn:parity-check-equations} can be written as a system of binary linear equations of the segment parities \(p^t\):
\begin{equation}\label{eqn:system}
\begin{aligned}
    s_u 
    &= \bigoplus_{i=1}^n H_{u, i} z_i 
    = \bigoplus_{(v_1, \dots, v_\nc) \in \fbinary^\nc} \bigoplus_{i=1}^{n_\segis{\bima{v_1 \dots v_\nc}}} v_u z_{\segis{\bima{v_1 \dots v_\nc}}, i} \\
    &= 
    \bigoplus_{\substack{
        (v_1,...,v_{\nc}) \in \fbinary^{\nc}: \\
        v_u = 1
    }} p_\segis{\bima{v_1 \dots v_\nc}},
\end{aligned}
\end{equation}
for each \(u \in [\nc]\).
In the second equality, we group terms segment-wise and use the fact that \(H_{u, i} = v_u\) for all bit positions \(i\) belonging to segment \(\bima{v_1 \dots v_u \dots v_\nc}\). 

For even-weight codes, we extend the tree structure by an additional all-one row  (see gray row in Fig.~\ref{fig:tree-structure}). To maintain a unified notation for both standard and extended tree structure, we denote the syndrome associated with this row by \(s_0:=\bigoplus_{i=1}^n y_{\text{hd},i}\). Without increasing the number of required segments, the all-one row introduces an additional constraint for the noise effect generation:
\begin{equation}\label{eqn:system3}
    s_0 = \bigoplus_{(v_1,...,v_{\nc}) \in \fbinary^{\nc}} p_\segis{\bima{v_1 \dots v_\nc}}
\end{equation}
which ensures that each generated noise effect will lead to a candidate codeword with even \ac{HW}.

We call a vector of segment parities \(p^t=(p_\segis{0}, \dots, p_\segis{t-1})\) that fulfill \eqref{eqn:system} (and \eqref{eqn:system3} in the case of the extended structure) a \emph{valid parity configuration} for a given syndrome \(s^{n-k}\) and parity check matrix \(H\). Accordingly, a vector of \acp{HW} \(w_{\text{H}}^t = (w_{\text{H},\segis{0}}, \dots, w_{\text{H}, \segis{t-1}})\) whose parities form a valid parity configuration is called a \emph{valid \ac{HW} configuration}. By construction, any noise effect \(z^n\) whose subeffects \acp{HW} form a valid \ac{HW} configuration automatically fulfills \(H_{1:\nc, :} z^n = (s_1,\ldots,s_\nc) = s_1^\nc\) (or \(H_{0:\nc, :} z^n = (s_0,s_1,\ldots,s_\nc) = s_0^\nc\) for the extended structure).

\begin{algorithm}[tbp]
\caption{\texttt{UnitParityComp} \label{alg:comp-unit-parity}}

\KwIn{Syndrome bits \(s^{n-k}\); Non-unit parities \(p_{\text{nu}}^{t-\nc-1}\)}
\KwOut{Valid unit parities \(p_{\text{u}}^\nc\) given \(p_{\text{nu}}^{t-\nc-1}\), and \(s^{n-k}\).}
Compute unit parities \(p_{\text{u}}^\nc = (p_{(\bima{10\dots0})}, p_{(\bima{01\dots0})}, \dots)\):
\begin{equation}\label{eqn:system2}
\begin{aligned}
    p_\segis{\bima{100\dots0}} 
    &\gets s_1 \oplus 
    \bigoplus_{\substack{
        (v_2,...,v_{\nc}) \in \\ 
        \{0, 1\}^{\nc-1} \setminus \{0^{\nc-1}\}
    }}
    p_\segis{\bima{1 v_2 v_3 \dots v_\nc}}  \\
    p_\segis{\bima{010\dots0}} 
    &\gets s_2 \oplus 
    \bigoplus_{\substack{
        (v_1,v_3,v_4,\dots,v_{\nc}) \in \\ 
        \{0, 1\}^{\nc-1} \setminus \{0^{\nc-1}\}
    }}
    p_\segis{\bima{v_1 1 v_3 \dots v_\nc}}  \\
    \vdots\\
    p_\segis{\bima{00\dots01}} 
    &\gets s_\nc \oplus 
    \bigoplus_{\substack{
        (v_1, \dots ,v_{\nc-1}) \in \\ 
        \{0, 1\}^{\nc-1} \setminus \{0^{\nc-1}\}
    }}
    p_\segis{\bima{v_1 \dots v_{\nc-1}1}} 
\end{aligned}
\end{equation}
\vspace{-\baselineskip}
\BlankLine
\Return{\(p_{\text{u}}^\nc\)}\;
\end{algorithm}

To efficiently iterate through valid parity configurations, we observe that the parities of the segments whose binary representation contains a single \(1\), i.e., the unit vectors, appear in exactly one equation of \eqref{eqn:system}.
Thus, we can explicitly solve for these parities by isolating them on the left-hand side as shown in \eqref{eqn:system2} in Alg.~\ref{alg:comp-unit-parity}.
We refer to the parities on the left-hand side as \emph{unit} parities 
\(
    p_{\text{u}}^\nc = \big(p_{(\bima{v_1, \dots, v_\nc})} : \sum_{i=1}^n v_i = 1\big),
\) 
and those on the right-hand side as \emph{non-unit} parities 
\(
    p_{\text{nu}}^{t-\nc-1} = \big(p_{(\bima{v_1, \dots, v_\nc})} : \sum_{i=1}^n v_i > 1\big).
\)
Note that \(p_{(0)}\) is not part of \eqref{eqn:system2} and neither a unit nor non-unit parity. 
Together, \(p_{\text{u}}^\nc\), \(p_{\text{nu}}^{t-\nc-1}\), and \(p_{(0)}\) partition the parity configuration
\(
    p^t = \texttt{Concatenate}(p_{(0)}, p_{\text{nu}}^{t-\nc-1}, p_{\text{u}}^\nc),
\)
where $\texttt{Concatenate}$ combines the subvectors into the full parity configuration.

Given a set of known non-unit parities \(p_{\text{nu}}^{t-\nc-1}\) and the received syndrome \(s^{n-k}\), Alg.~\ref{alg:comp-unit-parity} uniquely determines the unit parities \(p_{\text{u}}^\nc\) that in combination with \(p_{\text{nu}}^{t-\nc-1}\) form a valid parity configuration. To iterate over all valid parity configurations, we can therefore iterate over all combinations of the non-unit parities \(p_{\text{nu}}^{t-\nc-1}\) and compute the valid unit parities via Alg.~\ref{alg:comp-unit-parity}, which can be efficiently implemented as a binary matrix multiplication over  \(\fbinary\). 

Once \(p_{\text{nu}}^{t-\nc-1}\) is known it remains to determine \(p_{(0)}\). Let \(\mathcal{P}_{(0)} \subset \fbinary\) denote the feasible values of \(p_{(0)}\) that, together with \(p_{\text{u}}^\nc\) and \(p_{\text{nu}}^{t-\nc-1}\), form a valid parity configuration. The computation of \(\mathcal{P}_{(0)}\) is given in Alg.~\ref{alg:comp-zero-parity}:
\begin{algorithm}[tbp]
\caption{\texttt{ZeroParityComp} \label{alg:comp-zero-parity}}
\KwIn{Unit parities \(p_{\text{u}}^\nc\); Non-unit parities \(p_{\text{nu}}^{t-\nc-1}\); Global parity \(s_0\)}
\KwOut{\(\mathcal{P}_{(0)}\): Set of valid values for the \(0\)-th parity \(p_{(0)}\) given \(p_{\text{u}}^\nc, p_{\text{nu}}^{t-\nc-1}\), and \(s_0\)}
\uIf{\(H\) is in Extended Tree Structure}{
    \Return{\(
        \{
            s_0 \oplus 
            \bigoplus\limits_{\substack{
                (v_1, \dots,v_{\nc}) \in \\ 
                \{0, 1\}^{\nc} \setminus \{0^{\nc}\}
            }}
            p_\segis{\bima{v_1 \dots v_\nc}}\}
        \}
    \)}
    \;
}
\uIf{\(H\) is in Standard Tree Structure}{
    \Return{
    \(
        \{0, 1\}
    \)}\;
}
\end{algorithm}
the standard tree structure does not enforce any constraint on \(p_{(0)}\), and we therefore additionally iterate over both values: \(\mathcal{P}_{(0)} = \{0, 1\}\). In extended tree structure, \(p_{(0)}\) is constrained by \eqref{eqn:system3} which can be solved for \(p_{(0)}\) with \(p_{\text{u}}^\nc, p_{\text{nu}}^{t-\nc-1}\) and \(s_0\) on the right hand side. 

\subsection{Decoding Algorithm}\label{sec:dec-alg-seg-grand}
To leverage the segment parity constraints in \eqref{eqn:system} and \eqref{eqn:system3}, we follow Segmented GRAND~\cite{rowshanJournalNew2025} and, at each time, generate one subeffect for each segment, followed by combining the subeffects into one complete noise effect. Let \(|\tilde{L}_{\segi,i}|\) denote the random variable of the reliability of the \(i\)-th bit in segment \(j\) and let \(|\tilde{\ell}_{\segi, i}|\) be its realization. As pointed out by~\cite{ahsanullah2013introduction}%
\footnote{See Example 10.3 in \cite{ahsanullah2013introduction} for more detail.}, 
the q-th sample quantile of the reliabilities in every segment \(j\), \(|\tilde{L}_{\segi,1}|, |\tilde{L}_{\segi,2}|, \dots, |\tilde{L}_{\segi,n_\segi}|\), converges to \(F^{-1}(q)\) as $n_\segi\rightarrow\infty$, where \(F\) is the cumulative distribution function (CDF) of \(|\tilde{L}_{\segi,i}|\) and
\(F^{-1}(s) = \inf\{u: F(u) \geq s\}\). As a result, we expect that the plot of rank-ordered reliabilities \(|L_{\segi, i}|\) versus the rank \(i\) is approximately the function 
\begin{equation}\label{eqn:approx}
    |L_{\segi, i}| \approx F^{-1}(i/n_\segi), \qquad \text{for \(i \in [n_\segi]\)},
\end{equation}
where we denote rank-ordered reliabilities without tilde. Duffy et al.~\cite{duffyOrderedReliabilityBits2022} observed that for practical noise effect generation, the rank-ordered reliabilities can be approximated by a line with intercept \(0\). Applying a first-order Taylor approximation to \eqref{eqn:approx} with intercept \(0\) results in a linear relation between rank and reliabilities
\[
    \begin{aligned}
    |L_{\segi,i}| &\approx F^{-1}(i/n_\segi) \approx \beta_\segi i, \\
    \quad \beta_\segi &= \frac{1}{n_\segi} \frac{d F^{-1}(s)}{d s} = \frac{(F^{-1})^\prime}{n_\segi}
    \end{aligned}
\]
for \(i \in [n_\segi]\) and  \(j \in \{0, \dots, t-1\}\), where \(\beta_\segi\) denotes the slope of the \(j\)-th segment.
Since the tree structure is balanced (i.e., \(n_\segi \approx n / t\)) and \(F\) is segment independent, the slopes \(\beta_\segi\) are all approximatly equal: \(\beta_\segis{0} \approx \dots \approx \beta_\segis{t-1}\).
With this approximation, the reliability of a noise effect \(z^n\) in \eqref{eqn:reliabilities} for a realization \(\ell_{\segi,i}\) becomes
\begin{equation}\label{eqn:rel-gen-seg-grand}
\begin{aligned}
    \text{Rel}(z^n) 
    &= 
    \sum_{j=0}^{t-1} \sum_{i=1}^{n_\segi} |\ell_{\segi, i}| z_{\segi, i}
    \approx \beta_\segis{0} w_\text{T}.
\end{aligned}
\end{equation}
Owing to the balanced property of the tree structure, we can factor the common slope \(\beta_\segis{0}\), which makes the reliability of \(z^n\) proportional to its total \ac{LW} \(w_\text{T}\). Combined with \eqref{eqn:reliabilities} and \eqref{eqn:rel-gen-seg-grand}, this shows that the probability of a noise effect decreases monotonically with increasing total \ac{LW} \(w_\text{T}\).
Consequently, constrained noise effects with equal probability can be generated by merging independently generated subeffects that share a valid \ac{HW} configuration and whose segment-wise \acp{LW} sum to the same total \(w_\text{T}\).

Algorithm~\ref{alg:noise-generator} describes the high-level components of the \ac{GSegGRAND} algorithm: an outer iteration iterates over increasing total \ac{LW} \(w_\text{T}\) to create noise effects in decreasing order of probability. For each \(w_\text{T}\), the algorithm generates valid \ac{HW} configurations \(w_{\text{H}}^t\) corresponding to the received syndrome \(s^{n-k}\).
Specifically, the \ac{HW} generator \(\texttt{HWGen}(w_\text{T})\) iterates over all possible \ac{HW} assignments \(w_{\text{H,nu}}^{t-\nc-1}\) for the non-unit segments, for which the current total \ac{LW} \(w_\text{T}\) can be achieved (see the implementation details in Sec.~\ref{sec:implementation} for details). For each \(w_{\text{H,nu}}^{t-\nc-1}\), \texttt{UnitParityComp} and \texttt{ZeroParityComp} given in Alg.~\ref{alg:comp-unit-parity} and Alg.~\ref{alg:comp-zero-parity}, respectively, determine the parities of the unit segments and of segment \(0\) that result in a valid \ac{HW} configuration. For these segments, \(\texttt{HWGen}(w_\text{T}, p_\text{u}^\nc)\) and \(\texttt{HWGen}(w_\text{T}, p_\segis{0})\) only iterate over the \acp{HW} with correct parity \(p_\text{u}^\nc\) and \(p_\segis{0}\), skipping every second \ac{HW} for each segment. For each valid \ac{HW} configuration \(w_{\text{H}}^t\), the 2D Landslide algorithm \(\texttt{2DSegLS}(w_\text{T}, w_\text{H}^t)\) (see Sec.~\ref{sec:implementation}) generates all noise effects of this \ac{HW} configuration \(w_{\text{H}}^t\) and \ac{LW} \(w_\text{T}\) by partitioning the \ac{LW} among the segments.

\SetKwFunction{UnitParityComp}{UnitParityComp}
\SetKwFunction{ZeroParityComp}{ZeroParityComp}
\SetKwFunction{HWGen}{HWGen}
\SetKwFunction{Concatenate}{Concatenate}
\SetKwFunction{TwoDSegLS}{2DSegLS}

\begin{algorithm}[tbp]
\caption{\ac{GSegGRAND} Noise Effect Generator}\label{alg:noise-generator}
\small
\tcc{Iterate over increasing total \ac{LW} \(\Leftrightarrow\) decreasing noise effect prob.}
\For{\(w_\textnormal{T} \in \{0, 1, 2, \dots\}\)}{
    \tcp{Iterate over non-unit \acp{HW}}
    \For{\(
        w_{\textnormal{H,nu}}^{t-\nc-1} \in \HWGen{\ensuremath{w_\textnormal{T}}}
    \)}{
        \tcc{Iterate over unit \acp{HW} and incorporate \(H_{1:\nc, :} z^n = s_1^\nc\).}
        \(p_{\text{nu}}^{t-\nc-1} \gets w_{\textnormal{H,nu}}^{t-\nc-1} \bmod 2\)\;
        \(p_\textnormal{u}^\nc 
            \gets 
            \UnitParityComp{
                \ensuremath{p_{\textnormal{nu}}^{t-\nc-1}, s^{n-k}}
            }
        \)\;
        \For{\(
            w_{\textnormal{H,u}}^{\nc} \in \HWGen{\ensuremath{w_\textnormal{T}, p_\textnormal{u}^\nc}}
        \)}{
            \tcc{Iterate over segment \(0\) \acp{HW} and incorporate code evenness}
            \(\mathcal{P}_0 
            \gets  
            \ZeroParityComp{
                \ensuremath{p_{\text{u}}^\nc, p_{\textnormal{nu}}^{t-\nc-1}, s_0}}
            \)\;
            \For{$p_\segis{0} \in \mathcal{P}_0$, {\normalfont\bfseries for} $w_{\textnormal{H},\segis{0}} \in \HWGen{\ensuremath{w_\textnormal{L}, p_\segis{0}}}$}{
                \(w_\textnormal{H}^t \gets \texttt{Concatenate}(w_{\textnormal{H,\segis{0}}}, w_{\textnormal{H,nu}}^{t-\nc-1} \!\!, w_{\textnormal{H,u}}^{\nc})\)\;
                \For{\(
                    z^n \in \TwoDSegLS{\ensuremath{w_{\textnormal{T}}, w_\textnormal{H}^t}}
                \)}{
                    \textbf{yield} \(z^n\)\;
                }
            }
        }
    }
}
\end{algorithm}

\subsubsection{Example: Even-weight Code with 3 Constraints}
Figure~\ref{fig:example-generator} visualizes this process for an even-weight code and \(3\) constraints. In this case, there exist only one non-unit segment \(j=\bima{11}\) whose HW can be assigned freely. For each \(w_{\text{H}, \segis{\bima{11}}}\), the parities of the unit segments follow directly as:
\begin{align*}
    p_{\segis{\bima{10}}} = s_2 \oplus p_\segis{\bima{11}}, \quad
    p_{\segis{\bima{01}}} = s_1 \oplus p_\segis{\bima{11}}
\end{align*}
where \(p_\segis{\bima{11}} = w_{\text{H},\segis{\bima{11}}} \bmod 2\), allowing the \ac{HW} generator to skip every second \ac{HW} with incorrect parity. Likewise, the \ac{HW} generator for the \(0\) segment restricts generation of subeffects to parities of
\begin{align*}
    p_\segis{\bima{00}} &= s_0 \oplus p_\segis{\bima{10}} \oplus p_\segis{\bima{01}} \oplus p_\segis{\bima{11}}.
\end{align*}
For a given set of \acp{HW}, \texttt{2DSegLS} produces subeffects with corresponding \ac{HW} whose \acp{LW} sum to \(w_\text{T}\), which are merged and sent to the codebook checker.

\begin{figure}
    \centering
    \includegraphics[width=\linewidth]{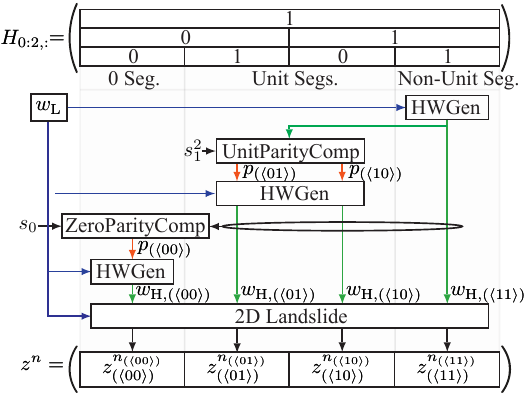}
    \caption{Visualization \ac{GSegGRAND} noise effect generator for an even weight code with three constraints: the noise effect generator iterates over the \acp{HW} of the non-unit (NU) segments in an outer iteration, with iterations over the \ac{HW} of the unit segments and \(0\) segment in an inner loop.}
    \label{fig:example-generator}
\end{figure}

\subsection{Balanced Tree Transformation}\label{sec:bal-tree-transf}
As outlined in Sec.~\ref{sec:dec-alg-seg-grand}, the decoder operates on an approximately balanced tree structure, i.e., the segments have approximately equal length. In this section, we propose a transformation which we call \emph{balance tree transformation} that approximately balances parity check matrices. The idea of this transformation is to apply a random invertible matrix to the parity check matrix $H \in \fbinary^{m \times n}$, where \(m \coloneqq n - k\), before rewriting it in tree structure. The probability that we get a transformed matrix with the desired property is bounded below by $0.22$ as $m\rightarrow \infty$. To provide a theoretical validation, Theorem~\ref{thm_uniform_matrix} is first proved without invertible assumption on the random transformation. After Theorem \ref{thm_uniform_matrix}, we will restrict to invertible random transformation so that the rewritten code is equivalent to the given code.
\begin{thm}\label{thm_uniform_matrix}
    Let $H^{m,\kappa m}$ be a given $m\times \kappa m$ binary matrix, for $\kappa\in \mathbb{N}$. Separate $H^{m,\kappa m}$ into $\kappa$ non-overlapping sub-matrices with dimension $m \times m$ and denote each of the sub-matrices as $H^{m,\kappa m}_{\{j\}}$. Let 
    \[
        \mathcal{H}^{m,\kappa m}:= \{H^{m,\kappa m}|H^{m,\kappa m}_{\{j\}} \text{ is invertible for all }j\leq \kappa \}
    \]
    be the set of parity check matrices whose sub-matrices are all invertible. Let $\widetilde{H}^{m,\kappa m}=A^{m,m}H^{m,\kappa m}$, where $A^{m,m}\in\fbinary^{m\times m}$ is a random binary matrix whose entries are i.i.d Bernoulli distributed with rate $0.5$. Then,
    \begin{enumerate}[label=\arabic*), leftmargin=*]
        \item Each $\widetilde{H}^{m,\kappa m}_{\{j\}}=A^{m,m}H^{m,\kappa m}_{\{j\}}$ is uniformly distributed over all binary $m\times  m$ matrices given $H^{m,\kappa m}\in\mathcal{H}^{m,\kappa m}$ and hence can be treated as matrix with i.i.d entries from Bernoulli distribution with rate $0.5$. 
        \item The rows in $\widetilde{H}^{m,\kappa m}$ are independent to each other. 
        \item Given $\delta>0$, let $\{H^{m,\kappa m}\}_{m\in\mathbb{N}}$ be any given sequence of parity check matrices such that $H^{m,\kappa m}\in\mathcal{H}^{m,\kappa m}$ for all $m$, and let $H^{m,\kappa m}_{{\{j\}},i,s}$ denote the entry of $H^{m,\kappa m}_{\{j\}}$ at the $i$-th row and $s$-th column, then 
        \begin{align}
            \lim_{m\rightarrow \infty}P\left(\bigcup_{j=1}^{\kappa}\bigcup_{i=1}^m\Bigg\{\left|\sum_{s=1:m}\frac{\widetilde{H}^{m,\kappa m}_{\{j\},i,s}}{m}-0.5\right|>\delta\Bigg\}\right)=0.
        \end{align}
    \end{enumerate}
\end{thm}

\begin{proof}
    1) Let $\widetilde{h}\in\fbinary^{m\times m}$ be any matrix. Then,
    \begin{align*}
       &P\left(\widetilde{H}^{m,\kappa m}_{\{j\}}
       =
       \widetilde{h}|H^{m,\kappa m}_{\{j\}}\right) 
       \stackrel{(\text{a})}{=}
       P\left(A^{m, m}H^{m,\kappa m}_{\{j\}}=\widetilde{h}|H^{m,\kappa m}_{\{j\}}\right)\\
       &\stackrel{(\text{b})}{=}
       P\left(A^{m,m}=\widetilde{h}[H^{m,\kappa m}_{\{j\}}]^{-1} \middle|H^{m,\kappa m}_{\{j\}}\right)=2^{-m^2}.
    \end{align*}
    (a) uses the invertibility of the sub-matrix \(H^{m,\kappa m}_{\{j\}}\) as $H^{m,\kappa m}\in\mathcal{H}^{m,\kappa m}$. (b) holds since the \(m^2\) entries of \(A^{m,m}\) are i.i.d. Bernoulli distributed with rate \(0.5\) so it equals \(\widetilde{h}[H^{m,\kappa m}_{\{j\}}]^{-1}\) with probability \((1/2)^{m^2}\).
    Hence $\widetilde{H}^{m,\kappa m}_{\{j\}}=A^{m, m}H^{m,\kappa m}_{\{j\}}$ is uniformly distributed over all binary $m\times  m$ binary matrices. With coupling technique \cite{Thorisson2000CouplingSA}, $\widetilde{H}^{m,\kappa m}_{\{j\}}$ can be treated as a matrix with i.i.d entries from Bernoulli distribution with rate $0.5$. This completes the first claim.

    2)
    Notice that $\widetilde{H}^{m,\kappa m}_{i,j}=\sum_{s=1}^m A^{m,\kappa m}_{i,s}H^{m,\kappa m}_{s,j} \mod 2$. Since $H^{m,\kappa m}$ is given and the rows of $A^{m,\kappa m}$ are independent, the rows of $\widetilde{H}^{m,\kappa m}$ are also independent. This completes the second claim.

    3)
    To prove the third claim, we consider the sum 
    \(S_{\{j\}, i, m} = \sum_{s=1}^m \widetilde{H}^{m,\kappa m}_{\{j\},i,s} / m\) 
    whose summands are bounded between \(0\) and \(1/m\), and according to claim 1, i.i.d. Bernoulli-distributed with rate \(0.5\). Hence, the expected value of \(S_{\{j\}, i, m}\) is \(0.5\) and applying Hoeffding's inequality~\cite{hoeffding1963probability} to \(S_{\{j\}, i, m}\) yields, given \(\delta > 0\),
    \(
        P(|S_{\{j\}, i, m} - 0.5| > \delta) 
        \leq 
        2 \exp(-2\delta^2 m).
    \)
    Applying the union bound to \(\bigcup_{i=1}^m\{|S_{\{j\}, i, m} - 0.5| > \delta\}\) results in
    \begin{align*}
      \lim_{m\rightarrow \infty}P\left(\bigcup_{i=1}^m
      \Bigg\{
          \left|\sum_{s=1}^m\frac{\widetilde{H}^{m,\kappa m}_{\{j\},i,s}}{m}-0.5\right| > \delta
      \Bigg\}
      \right)=0,
    \end{align*}
    because \(\lim_{m \to \infty} \sum_{i=1}^m 2\exp(-2\delta^2m) = 0\).
    This means we can expect that a sub-matrix has almost half the number of ones and zeros within each row. By the union bound, we are able to conclude that this happens to all sub-matrices of $\widetilde{H}^{m,\kappa m}$ simultaneously as we have a fixed code rate $1/\kappa$, which finishes the third claim. 
\end{proof}
The first result of Theorem \ref{thm_uniform_matrix} states that a parity check matrix, whose sub-matrices are invertible, can be transformed into a random matrix whose entries in each sub-matrix can be treated as i.i.d. variables with equal probability of being one or zero. This further leads to the properties stated in the second and third results. Choosing $\delta$ close to zero and applying the third result shows that as $m\rightarrow\infty$, we can expect that the number of ones is approximately the same as the number of zeros in each row of $\widetilde{H}^{m,\kappa m}$, while the coderate is constant: \(r = 1 - 1 / \kappa\). Together with the fact that each row is independent in $\widetilde{H}^{m,\kappa m}$, we can expect that each segment length \(n_\segis{\bima{v_1 \dots v_\nc}}\) of $\widetilde{H}^{m,\kappa m}$ in tree structure is approximately the same, that is, the set of columns are divided into subsets with approximately equal sizes. The name balanced tree transformation is given based on this property.

To retain the same structure of the code represented by $H^{m,\kappa m}$, we further require that the random transformation $A^{m,m}$ is invertible. The probability that a random binary matrix $A^{m, m}$ is invertible is bounded below by 0.22 as $m\rightarrow \infty$ \cite{BRENNAN1987311,waterhouse1987often,kolchin1999random}. Therefore, by the union bound, the probability that $A^{m,m}$ is invertible and each segment length \(n_\segi\) of $\widetilde{H}^{m,\kappa m}$ in tree structure is approximately the same is bounded below as $m\rightarrow \infty$. This establishes that, in the limit, we can transform parity check matrices into tree structure with desired properties with positive probability bounded below. We call the procedure of multiplying $H$ by a random matrix the \textit{balance transformation}. Additionally, the procedure of balance transformation followed by rewriting the resulting matrix into tree structure is called \textit{balanced tree transformation}. 

As the balanced tree transformation is invertible and hence does not change the code, we only need to apply the transformation once for any given parity check matrix and use the rewritten parity check matrix and generator matrix for application. Therefore, this transformation will not add additional complexity to the decoding.
 
\subsection{Implementation Details}\label{sec:implementation}
An important implementation detail is that the \ac{HW} configurations are restricted to those configurations for which the given total \ac{LW} \(w_\text{T}\) is realizable; otherwise, the algorithm would generate configurations for which 2D Landslide cannot find valid noise effects. This is particularly important as \ac{GSegGRAND} supports multiple parity check constraints requiring numerous segments, which can lead to a large number of invalid \ac{HW} configurations if not handled properly. To obtain bounds efficiently, \ac{GSegGRAND} places the parity-check constraints in its outer loops and performs logistic-weight partitioning within its inner loops. This is in contrast to Segmented GRAND, which first partitions the \ac{LW} before enforcing parity constraints. Furthermore, the iteration over error pattern bases (segments without bit flips) used in \cite{rowshanJournalNew2025} is implicitly incorporated in the proposed approach: a segment with \ac{HW} \(0\) naturally corresponds to an empty segment, eliminating the need for an explicit base iteration.

This iteration order provides an additional advantage: because the structure of the proposed algorithm on the segment level corresponds to \ac{ORBGRAND}, i.e., first \ac{HW} iteration followed by \ac{LW} partitioning, the decoder can build directly upon the existing Landslide algorithm, for which efficient hardware implementation already exists~\cite{riazSub08pJBitUniversal2025}. The following section outlines how the outer \ac{HW} generator and inner 2D Landslide algorithm are realized.

\subsubsection{Hamming Weight Generator (\texttt{HWGen})}
The \ac{HW} generation in Alg.~\ref{alg:noise-generator} is executed through multiple stages to incorporate the parity check constraints. Note that for the $j$-th segment, a subeffect with \ac{HW} \(w_{\text{H},\segi}\) can only have \acp{LW} \(w_{\text{L},\segi}\) within the range, 
\(w_{\text{L,min},\segi} \leq w_{\text{L},\segi} \leq w_{\text{L,max},\segi}\), where~\cite{duffyOrderedReliabilityBits2022}
\[
    \begin{aligned}
        w_{\text{L,min}, \segi} &= \frac{w_{\text{H},\segi} (w_{\text{H},\segi} + 1)}{2}, \\
        w_{\text{L,max}, \segi} &= \frac{n_\segi (n_\segi + 1)}{2} \\
        &\phantom{=} - \frac{(n_\segi - w_{\text{H},\segi}) (n_\segi - w_{\text{H},\segi} + 1)}{2}.
    \end{aligned}
\]
A given \ac{HW} configuration \(w_\text{H}^t\) therefore has noise effects with \ac{LW} within the sum of the segment-wise bounds 
\begin{equation}\label{eqn:bounds-log-weight}
    \sum_{j=0}^{t-1} w_{\text{L,min},\segi}
    \leq
    w_\text{T}
    \leq
    \sum_{j=0}^{t-1} w_{\text{L,max},\segi}.
\end{equation}
2D Landslide can find valid noise effects if the current total \ac{LW} \(w_\text{T}\) is within these bounds. We ignore the upper bound, as \ac{GSegGRAND} most of the time terminates for a relatively small \(w_\text{T}\) and the bound does not restrict \(w_\text{H}^t\) in practice. To incorporate the lower bound, assume that the \ac{HW} generator has already chosen the \ac{HW} for a subset of segments \(j_1, \dots, j_r\). The lower bound of \eqref{eqn:bounds-log-weight} enforces an upper bound on the \ac{HW} of the next segment \(j_{r+1}\)
\begin{equation}\label{eqn:cond}
    \frac{
        w_{\text{H},\segis{j_{r+1}}} 
        (w_{\text{H},\segis{j_{r+1}}} \hspace{-0.2em} + \hspace{-0.15em} 1)
    }{2} 
    = w_{\text{L,min},\segis{j_{r+1}}} \leq w_\text{T} - \sum_{i=1}^{r} w_{\text{L,min},\segis{j_i}}.
\end{equation}

In practice, the \ac{HW} generator increases a \ac{HW} \(w_{\text{H},\segis{j_{r+1}}}\) starting from \(0\) as long as \eqref{eqn:cond} is fulfilled. As soon as \eqref{eqn:cond} is violated, \(w_{\text{H},\segis{j_{r+1}}}\) is set to \(0\) and the generator increases the previous \ac{HW} \(w_{\text{H},\segis{j_{r}}}\) by one. If the increase also violates \eqref{eqn:cond} it moves up another level. If the \ac{HW} \(w_{\text{H}, \segis{j_r}}\) corresponds to a unit segment, the generator uses the same stopping condition but starts at the precalculated parity \(p_\segis{j_r}\) instead of \(0\) and increases the \ac{HW} by \(2\) each step. This ensures that the generator implicitly only generates \acp{HW} of correct parity.

\subsubsection{2D Landslide (\texttt{2DSegLS})}
The 2D Landslide algorithm \(\texttt{2DSegLS}(w_\text{T}, w_{\text{H}}^t)\) partitions the total \ac{LW} \(w_\text{T}\) among the segments: \(w_\text{L}^t = (w_{\text{L}, \segis{0}}, \dots, w_{\text{L}, \segis{t-1}})\) and generates noise effects whose subeffects have \ac{HW} \(w_{\text{H}, \segi}\) and \ac{LW} \(w_{\text{L}, \segi}\).
To realize this, the algorithm runs, for each integer partition \(w_\text{L}^t\), the Landslide algorithm (see Sec.~\ref{sec:orbgrand}) on each segment.
In practice, this forms a nested enumeration over all combinations of subeffects: each landslide generator \(\text{LS}(w_{\text{H}, \segi}, w_{\text{L}, \segi})\) produces its first noise effect and the last generator iterates through all its effects. When it runs out of effects, it resets, and the preceding landslide generator advances to its next noise effect. This process continues in a nested manner until the first generator has no remaining effects. After each enumeration, the subeffects are merged and tested by the codebook checker.

\section{Soft-output Calculation for Constrained Guessing}\label{sec:socalculation}
In this section, we outline how \ac{SOGRAND}'s \ac{SO} computation can incorporate constrained guessing, enabling us later to use \ac{GSegGRAND} as a component decoder for turbo product decoding with minimal performance loss (less than \SI{0.2}{\dB}) over \ac{ORBGRAND} without constrained guessing.
We explain first the core idea for noise effect skipping in \ac{SOGRAND}, which was proposed in \cite{fengLeveragingCodeStructure2025}, and then apply the idea to \ac{GSegGRAND}.

Constrained guessing restricts the noise effect generation to noise effects \(z^n\) to a subspace \(\const \subset \fbinary^n\). 
If a linear code is even, i.e., has only codewords of even weight, \ac{ORBGRAND} can restrict generation to noise effects of correct parity \(s_0 = \Phi(y_{\text{hd}}^n) = \Phi(e^n)\):
\(
    \const_{\text{ORB},s_0} = \{z^n : \Phi(z^n) = s_0\},
\)
by leveraging Landslides' ability to generate effects of specific \ac{HW}. Feng et al. noticed that in this case, the quality of \ac{SOGRAND}'s \ac{SO} can be improved by conditioning the noise effect probability terms in \eqref{eqn:blockso} on the parity constraint \(\const_{\text{ORB},s_0}\)~\cite{fengLeveragingCodeStructure2025}. 
That means for a general constraint space \(\const\), we have
\begin{equation}\label{eqn:cond-noise-prob}
    P_{E^n | Y^n, E^n \in \const}\Big(z^n | y^n\Big)
    = 
    \ind_{\{z^n \in \const\}} P_{E^n | Y^n}(z^n | y^n) / \psi,
\end{equation}
where
\(
    \psi \coloneqq P_{E^n | Y^n}(\const | y^n) = \sum_{z^n \in \const} P_{E^n | Y^n}(z^n | y^n)
\)
and \(\ind_{\{z^n \in \const\}}\) denotes the indicator function, which is \(1\) if \(z^n \in \const\) and \(0\) otherwise.
Replacing the unconditional probabilities \(P_{E^n | Y^n}(z^n | y^n)\) in \eqref{eqn:blockso} with \eqref{eqn:cond-noise-prob} and expanding the fractions by \(\psi\) results in the improved \ac{SO} expressions
\begin{equation}\label{eqn:block-so-updated}
\begin{aligned}
    \hat{P}_{C^n | Y^n}(\hat{c}^n | y^n) 
    &=
    \frac{
        P_{E^n | Y^n}\left(\hat{c}^n \oplus y_\text{hd}^n | y^n\right)
    }{
        \plist  + 
        (\psi - \pnoise) 2^{-(n-k - \lambda)}
    }, \\
    \hat{P}_{C^n | Y^n}(\mathcal{C} \setminus \mathcal{L} | y^n) 
    &=
    \frac{
        (\psi - \pnoise) 2^{-(n-k - \lambda)}
    }{
        \plist  + 
        (\psi - \pnoise) 2^{-(n-k - \lambda)}
    },
\end{aligned}
\end{equation}
where \(\hat{c}^n \in \mathcal{L}\) and $\lambda=1$ while \(\pnoise\) and \(\plist\) are given by \eqref{eqn:prob-acc}. 

In addition to \(\psi\), \eqref{eqn:block-so-updated} introduces the additional parameter \(\lambda \in \mathbb{N}\). If only the even code property is leveraged, $\lambda=1$ as mentioned. If more constraints are leveraged as with Segmented GRAND or \ac{GSegGRAND}, $\lambda$ should increase accordingly. To understand its purpose, we recall how \ac{SOGRAND} \ac{SO} calculation works: the optimal \ac{SO} calculation normalizes each codeword's a posteriori probability by the codebooks probability mass \(\sum_{\hat{c}^n \in \mathcal{C}} P_{E^n | Y^n}\left(\hat{c}^n \oplus y_\text{hd}^n | y^n\right)\). Since list decoding only reveals the probability mass \(\plist\) of the codewords in the list, \ac{SOGRAND} assumes that the remaining \(2^k - |\mathcal{L}|\) codewords are uniformly distributed among the \(2^n - \guessingIdx\) unexplored noise effects, where \(\guessingIdx\) denotes the number of explored noise effects.  This corresponds to an approximate probability mass of \(2^{-(n-k)}\) times the unexplored noise effect probability \(1 - \pnoise\).
In contrast, a \ac{GRAND} variant that incorporates \(\lambda\) parity check constraints only guesses within a set of \(2^{n - \lambda}\) noise effects while the number of unexplored codewords stays constant, resulting in the factor \(2^{-(n-k-\lambda)}\).

When the even code structure is leveraged as in \cite{fengLeveragingCodeStructure2025}, \(\lambda=1\) and \(\psi\) is equal to the a posteriori probability of the noise effect parity, which can be directly calculated via \cite[Lemma~1]{gallagerLowdensityParitycheckCodes1962a}
\begin{align*}
    \psi 
    &= P_{E^n | Y^n}(\const_{\text{ORB},s_0} | y^n)
    = P_{\Phi(E^n) | Y^n}(s_0| y^n) \\
    &= \frac{1}{2} \Big(1 + (-1)^{s_0} \prod_{i=1}^n (1-2 P_{E_i | Y_i}(1 | y_i))\Big),
\end{align*}
where \(S_0 = \Phi(E^n)\) is the random variable of the global parity.
While \cite{fengLeveragingCodeStructure2025} showed that this modification improves the \ac{SO} quality, it is unclear whether the better \ac{SO} quality affects \ac{TPC} decoding.
In this paper, we apply the modification to \ac{TPC} and demonstrate that the modification is essential for \ac{GRAND} component decoding with constrained guessing as the the original \ac{SO} calculation can lead to high error floors for certain codes or performance loses of up to \SI{0.5}{\dB}.

To apply this concept to Segmented \ac{GRAND} and \ac{GSegGRAND}, we need to find an efficient way to calculate the probability mass \(\psi = P_{E^n | Y^n}(\const | y^n)\) of all noise effects \(\const\) over which the algorithms carry out guesswork. Segmented \ac{GRAND} and \ac{GSegGRAND} with standard tree structure restrict their guessing space by \(\nc\) linearly independent constraints, and the extended tree structure restricts the space by \(\nc+1\) constraints. Hence, we set \(\lambda = \nc\) or \(\lambda = \nc + 1\) accordingly.

The constraint space of Segmented GRAND is given by
\[
     \const_{\text{Seg},s^{n-k}} = \{z^n : H_{1:\nc, :} z^n = s^{n-k}\} = \bigcap_{j=1}^q \const_{\text{Seg},\segi,s_j},
\]
where
\(
    \const_{\text{Seg},\segi,s_j} = \{z^n: \Phi(z_\segi^{n_\segi}) = s_j\}
\)
is the constraint associated with segment \(j\).
Since Segmented \ac{GRAND} requires parity check equations with nonoverlapping \(1\)s for its constraint space, and the channel is memoryless, the individual constraints \(\const_{\text{Seg},\segi,s_j}\) are stochastically independent events \(\{E^n \in \const_{\text{Seg},\segi,s_j}\}\), resulting in a probability for \(\const\) of
\[
    \psi = P_{E^n | Y^n}(\const_{\text{Seg},s^{n-k}} | y^n)
    =
    \prod_{j=1}^\nc
    P_{S_j | Y^n}\left(s_j| y^n\right),
\]
where \(S^{n-k} = H E^n\) is the syndrome random variable.
The a posteriori probability of the segment's parity can be directly calculated via \cite[Lemma~1]{gallagerLowdensityParitycheckCodes1962a}
\begin{multline}\label{eqn:segmented-parity-prob}
    P_{S_j | Y^n}\left(s_j| y^n\right) 
    = 
    P_{\Phi\left(E^{n_\segi}_\segi\right) | Y^n}\left(s_j| y^n\right) \\
    = \frac{1}{2} \Big(1 + (-1)^{s_j} \prod_{i=1}^{n_\segi} (1-2 P_{E_{\segi,i} | Y_{\segi,i}}(1 | y_{\segi,i}))\Big),
\end{multline}
where \(e_{\segi,i}\) and \(y_{\segi,i}\) denote the \(i\)-th position of the \(j\)-th segment of the correct noise effect and received value.

\ac{GSegGRAND}'s constraints space can be expressed by a linear system of parity check constraints (see Sec.~\ref{sec:bal-tree-struc}): \(\const_\text{GSeg} = \{z^n: \eqref{eqn:system}\}\) and \(\const_\text{GSeg} = \{z^n: \eqref{eqn:system}, \eqref{eqn:system3}\}\) for the standard and extended tree structure, respectively, whose solutions are valid parity check configurations \(p^t\). Let
\[
    \const_{\text{GSeg},p^t} \coloneqq \left\{z^n: \forall j \in \{0, \dots, t-1\}, \Phi\left(z_\segi^{n_\segi}\right) = p_\segi 
    \right\}
\]
denote the set of the noise effects whose subeffect parities form the parity configuration \(p^t\). Then, \(\const_{\text{GSeg},p^t}\) form a partition of \ac{GSegGRAND} constraint space \(\const_\text{GSeg} = \bigcup_{p^t: \eqref{eqn:system}, [\eqref{eqn:system3}]} \const_{\text{GSeg},p^t}\), where the union is over all valid parity check configurations \(p^t\) satisfying \eqref{eqn:system} and, for the extended tree structure, also \eqref{eqn:system3}. 
By the same independence argument used for Segmented \ac{GRAND}, the a posteriori probability of \(\const_{\text{GSeg},p^t}\) factorizes as
\[
    P_{E^n | Y^n}(\const_{\text{GSeg},p^t} | y^n)
    =
    \prod_{j=0}^{t-1}
    P_{\Phi\left(E^{n_\segi}_\segi\right) | Y^n}\left(p_\segi | y^n\right),
\]
whose factors can each be calculated via \eqref{eqn:segmented-parity-prob} by replacing \(s_j\) with \(p_\segi\). 
Since the sets \(\const_{\text{GSeg},p^t}\) form a partition of \(\const_\text{GSeg}\), we obtain
\begin{equation}\label{eqn:sum-prob}
    P_{E^n | Y^n}(\const_\text{GSeg} | y^n) = \sum_{p^t: \eqref{eqn:system}, [\eqref{eqn:system3}]} P_{E^n | Y^n}(\const_{\text{GSeg},p^t} | y^n).
\end{equation}
where the sum is over \(p^t\) satisfying \eqref{eqn:system} and, for the extended tree structure, also \eqref{eqn:system3}. 

The valid parity check configurations in the sum of \eqref{eqn:sum-prob} can be efficiently iterated using
\texttt{UnitParityComp} and \texttt{ZeroParityComp} given in Alg.~\ref{alg:comp-unit-parity} and Alg.~\ref{alg:comp-zero-parity}, respectively, (see Sec.~\ref{sec:bal-tree-struc}) and we can reuse the parity iterator framework on which \ac{GSegGRAND} is built.
Note that the partition by parity check configurations enables the computation of \(P_{E^n | Y^n}(\const_\text{GSeg} | y^n)\) by a small number of summations, avoiding the infeasible evaluation of all terms in \(\const\). Specifically, for the extended tree structure with \(\nc=1\) and \(\nc=2\) we can list and implement the summation explicitly as \(\nc=1\) has only one summand
\(
    p^2 \in \{(s_0 \oplus s_1, s_1)\},
\)
and \(\nc=2\) has only two summands
\begin{multline*}
    p^4 \in \{
        (s_0 \oplus s_1 \oplus s_2, s_2, s_1, 0), \\
        (s_0 \oplus s_1 \oplus s_2 \oplus 1, s_2 \oplus 1, s_1 \oplus 1, 1)
    \}.
\end{multline*}

\section{Results}\label{sec:results}
\ac{ORBGRAND} and \ac{GSegGRAND} function with any channel that outputs \ac{LLR}. To evaluate their performance in this section, we consider a binary-input additive Gaussian noise channel (BI-AWGN) that maps codewords \(c^n \in \mathcal{C}\) to the channel output \(y^n\) with \(y_i = (-1)^{c_i} + \gaussnoise_i\) for \(i \in [n]\). \(\gaussnoise_i\) are independent normal distributed noise samples with variance \(\sigma^2 = (2 \lEsNO)^{-1}\), where \(\lEsNO = r \lEbNO\) and \(r\) is the code rate, resulting in the \acp{LLR}: \(\ell_i = 2 y_i / \sigma^2\).

\subsection{Balancing}
Theorem~\ref{thm_uniform_matrix} analyzed balancing of parity check matrices for \(n \to \infty\). In this section, we demonstrate how balancing can be applied to concrete codes with finite \(n\), resulting in the balanced tree structures used for the numerical results. To balance a concrete parity check matrix \(H\), \(H\) is multiplied with invertible binary matrices \(A\) sampled uniformly at random, and the columns are permuted into tree structure. A transformation \(A H\) is accepted if the discrepancy \(\max_j |n_\segi - n / 2^t|\) between the segment lengths \(n_\segi\) and the optimal balancing is below a certain threshold. To obtain an extended balanced tree structure for even codes, we first add an additional all-one row as first row, apply Gaussian elimination and remove the resulting last all-zero row. Then, balancing is performed on all rows excluding the first one as described above. Note that the balancing only needs to be applied once offline and does not contribute to the decoding complexity. We observed that the decoding performance does not degrade for small discrepancies, especially when the code is long. For even-weight Hamming codes (e.g., the (32, 26) eBCH code in this work), by definition all rows of the parity check matrix can be represented in balanced tree structure and balancing is not required.

\begin{table}[t]
\caption{Accuracy balanced tree transformation for \(4\) rows}
\label{tab:balance}
\centering
\begin{tabular}{l r r}
\hline
Code & \makecell{Absolute \\ Discrepancy} & \makecell{Relative \\ Discrepancy} \\
\hline
(256, 239) eBCH  & 2 & 6.3 \% \\
(128, 106) eBCH  & 0 & 0 \% \\
(32, 21) dRM  & 0 & 0 \% \\
(32, 26) eBCH  & 0 & 0 \% \\
(128, 110) CA-Polar code & 1 & 6.3 \% \\
\hline
\end{tabular}
\end{table}

In this work, we demonstrate decoding with up to 8 segments and therefore balance up to \(4\) rows in extended tree structure. Table~\ref{tab:balance} lists the results of the balancing for \(4\) rows, and lists the absolute discrepancy \(\max_j |n_\segi - n / 2^t|\) and relative discrepancy \(\max_j |n_\segi - n / 2^t| / (n / 2^t)\): we achieve balancing with a maximal absolute discrepancy of \(2\) for (256, 239) eBCH code, \(1\) for the (128, 110) CRC-Assisted Polar (CA-Polar) code, and perfect balancing for all other codes.
Figure~\ref{fig:balancing-example} visualizes this process on the first \(4\) parity check rows of the (128, 106) eBCH code (matrix \(H\)), where orange and blue rectangles represents \(1\) and \(0\)s, respectively: the tree structure without rebalancing (\(H_2\)) has a strong imbalance. However, after balancing, the first \(4\) rows are perfectly balanced with \(8\) segments of lengths \(16\).

\begin{figure}
    \includegraphics{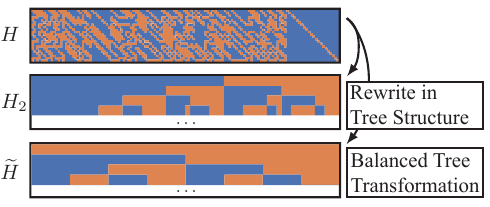}
    \caption{Example Balance Tree Transformation for the (128, 106) eBCH code}
    \label{fig:balancing-example}
\end{figure}

\subsection{Decoding of Component Codes}
We demonstrate the performance of \ac{GSegGRAND} on the (128, 106), and (256, 239) eBCH codes, which are visualized in Fig.~\ref{fig:ebch256-239}, and \ref{fig:ebch128-106}. In every legend entry, we denote the number of constraints that the decoder considers in brackets. Efficient decoding of the eBCH (256, 239) is particularly important as it serves as the component code of the OFEC code defined in the Open ROADM standard in optical communications~\cite{openroadmOpenROADMMSA2021}. \cite{rowshanJournalNew2025} originally tested Segmented GRAND on the (128, 106) eBCH demonstrating that the required guesswork for it can be reduced by up to a factor of \(4\) by incorporating \(2\) constraints. Following up on this, we demonstrate \ac{GSegGRAND}'s strength and incorporate up to \(4\) constraints.
\begin{figure}[tb]
    \centering
    \includegraphics{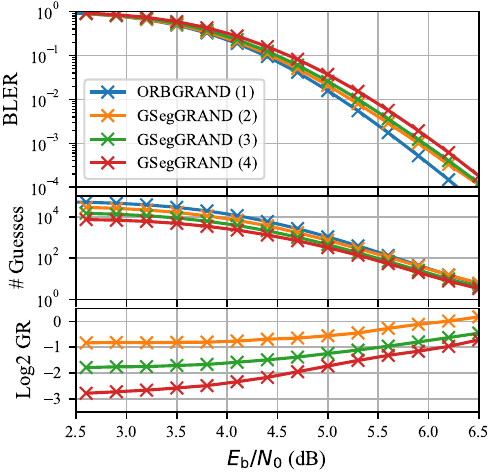}
    \caption{Decoding performance and guesswork for an (256, 239) eBCH with GSegGRAND vs. ORBGRAND decoding.}
    \label{fig:ebch256-239}
\end{figure}
\begin{figure}[tb]
    \centering
    \includegraphics{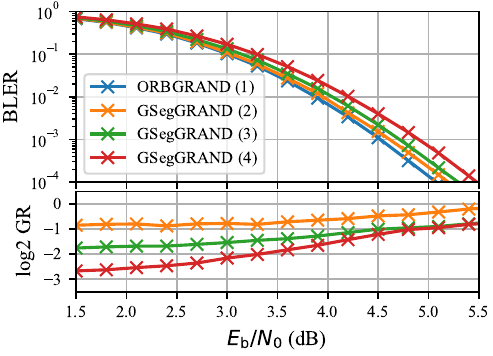}
    \caption{Decoding performance and guesswork for an (128, 106) eBCH with GSegGRAND vs. ORBGRAND decoding.}
    \label{fig:ebch128-106}
\end{figure}

To compare all decoders at their maximal performance, we use an abandonment threshold of \(2^{n-k}\). This threshold is the average number of guesses until an incorrect codeword is found, and further guessing does in general not improve performance~\cite{duffyCapacityAchievingGuessingRandom2019}. For \ac{GSegGRAND}, we use 2, 4, and 8 segments, which take, for even-weight codes, \(2\), \(3\), and \(4\) constraints into account, respectively. Our baseline is \ac{ORBGRAND} without segmentation, whose noise effect generator can incorporate \(1\) constraint if the code is even. We can therefore expect a maximal guesswork reduction of \(2^{-1}\), \(2^{-2}\), and \(2^{-3}\) compared to \ac{ORBGRAND}. 
The upper plot of each figure shows the \ac{BLER} vs \(\lEbNO\) and the middle plot in Fig.~\ref{fig:ebch256-239} shows the average guesswork in \(\log_{10}\)-scale. We define the average guesswork as the average number of noise effect generations carried out during decoding until a codeword is found or decoding is abandoned. For clarity, we calculate the \emph{log2 \ac{GR}}, which is the ratio between the decoder's and \ac{ORBGRAND}'s average guesswork in \(\log_2\)-scale and plot them in the lower plot of each figure.

\ac{GSegGRAND} experiences a slight performance loss at low \ac{SNR}, which is caused by the independent generation of subeffects through which the relative ordering between bits of different segments is lost. More specifically, the lost is caused by the difference between \eqref{eq_ORB_rel} and \eqref{eqn:rel-gen-seg-grand}. However, this loss is exchanged for a guesswork reduction of up to \(2^3=8\) at low \ac{SNR}, showing that \ac{GSegGRAND} utilizes the constraints fully.
This low \ac{SNR} regime is exactly turbo product decoding operates in early iterations. 
\begin{figure}
    \centering
    \includegraphics{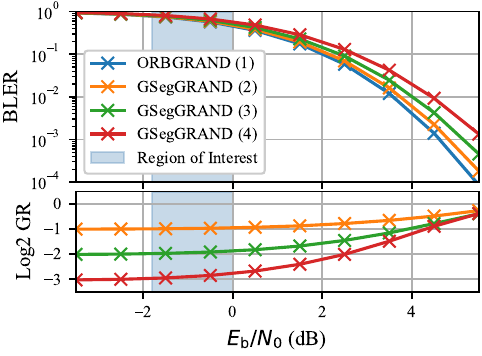}
    
    \caption{Decoding performance and guesswork for the (32, 21) dRM component code for corresponding \((32, 21)^2\) dRM \ac{TPC} (see Fig.~\ref{fig:tpc-dRM-32-21}) with GSegGRAND vs. ORBGRAND decoding. 
    The blue shaded region of interest marks the waterfall region of the corresponding \ac{TPC}. In the first \ac{TPC} decoding iteration, the component decoders operate within this \ac{SNR} range.}
    \label{fig:dRM32-21}
\end{figure}
Fig.~\ref{fig:dRM32-21} shows the \ac{BLER} and Log2 \ac{GR} of a (32, 21) \ac{dRM} code and marks the SNR range of the corresponding \ac{TPC}'s waterfall region%
\footnote{Concretely, we converted the \(\lEbNO\) range of the \ac{TPC}'s waterfall region from Fig.~\ref{fig:tpc-dRM-32-21} into a \(\lEsNO\) range and back into the corresponding \(\lEbNO\) range of the component code.}
(see Fig.~\ref{fig:tpc-ebch32_26} for decoding results of the corresponding \ac{TPC}).
In the \ac{TPC}'s waterfall region, \ac{GSegGRAND} achieves almost maximal guesswork reduction for the corresponding component code.

\subsection{Decoding of Turbo Product Codes}
We compare \ac{TPC} decoding with \ac{GSegGRAND} and ORBGRAND component code decoding. For all \acp{TPC}, we use the following decoding parameters as in the original \ac{SOGRAND} paper\cite{yuanSoftoutputGRANDLong2023}: damping factor \(\alpha = 0.5\) and list size \(N_{\mathcal{L}} = 4\). List decoding terminates when the estimated list error is below \(T = \num{e-5}\) or if \num{10000} guesses are reached. 
To specify a \ac{TPC}'s parameters and dimensions, we denote a \ac{TPC} based on an \((n,k)\) component code \(\mathcal{C}_\text{c}\) as ``\((n^2, k^2) = (n,k)^2 \; \mathcal{C}_\text{c}\) \ac{TPC}'' in the figures captions.

Fig.~\ref{fig:tpc-dRM-32-21} shows the decoding performance for a \((32, 21)^2\) \ac{dRM} \ac{TPC}, which outperformed its corresponding 5G LDPC code under \ac{SOGRAND} decoding with \ac{ORBGRAND}~\cite{yuanSoftoutputGRANDLong2023}. 
\begin{figure}
    \centering
    \includegraphics{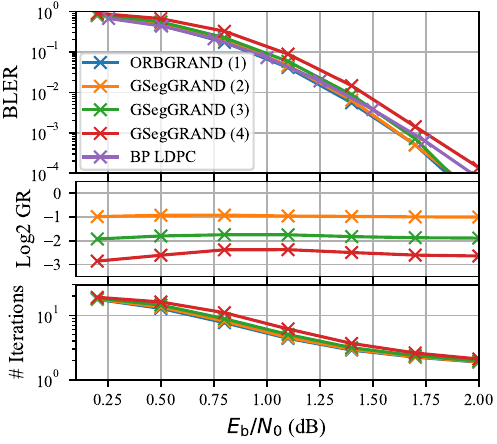}        
    \caption{Decoding performance and guesswork for the \((1024, 441) = (32, 21)^2\) dRM TPC. (The \ac{LDPC} \ac{BLER} curve is from~\cite{yuanSoftoutputGRANDLong2023}.)}
    \label{fig:tpc-dRM-32-21}
\end{figure}
The plot shows the \ac{TPC}'s \ac{BLER} performance, the Log2 \ac{GR}, and the average number of decoding iterations. The Log2 \ac{GR} is calculated from the total number of guesses, counting the guesses of all component decoders over all iterations until \ac{TPC} decoding terminates. The number of decoding iterations counts the number of full decoding iterations before termination.

Although (32, 21) dRM component decoding experienced small \ac{BLER} performance losses with \ac{GSegGRAND} in Fig.~\ref{fig:dRM32-21},  \ac{GSegGRAND} with 2 and 4 segments achieves \ac{ORBGRAND}'s \ac{TPC} decoding performance with almost the same average number of decoding iterations while reducing the guesswork by \(2\) and \(4\) across the \ac{SNR} range. The component decoder loss is most likely mitigated by the \ac{TPC} list decoding, which filters out small guessing order degradations. \Ac{GSegGRAND} with 8 segments reduces the guesswork further while experiencing a \ac{BLER} performance loss of less than \SI{0.2}{\dB}. In this case, each segment only contains 4 bit positions, resulting in a higher variance of the reliabilities per segment. \ac{TPC} decoding with \ac{GSegGRAND} therefore outperforms the corresponding \ac{LDPC} code under \ac{BP} decoding for 2 and 4 segments. To demonstrate generality, we also evaluated an  (1024, 676) \ac{TPC} with a (32,26) eBCH component code, for which we observe similar results, as shown in Fig.~\ref{fig:tpc-ebch32_26}.
\begin{figure}
    \centering
    \includegraphics{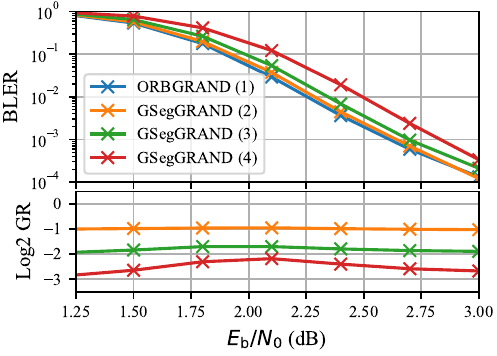}
    \caption{Decoding performance and guesswork for the \((1024, 676)=(32, 26)^2\) eBCH \ac{TPC}.}
    \label{fig:tpc-ebch32_26}
\end{figure}

\subsection{Soft-Output Calculation for Constrained Guessing}
To show the importance of the \ac{SO} calculation with conditioned noise effect probabilities \eqref{eqn:block-so-updated} in Sec.~\ref{sec:socalculation} for turbo product decoding, we compared its performance against the original \ac{SO} computation in \eqref{eqn:blockso} based on unconditional noise effect probabilities. We refer to the unconditioned and conditioned \ac{SO} computation as (SO) and (SO+), respectively.

We first compare how well the \ac{SO} of \ac{GSegGRAND} (SO) and (SO+) matches \ac{ORBGRAND}'s \ac{SO} without constrained guessing, which is our baseline in this section. To do so, we transmit \(10000\) (32, 26) eBCH codewords over a BI-AWGN channel with \(\lEbNO = \SI{3}{\dB}\)%
\footnote{Turbo product decoding with \ac{ORBGRAND} achieves a target \ac{BLER} of \num{e-4} at this operational point.}
and \ac{SISO}-decode the received signals with \ac{ORBGRAND} and \ac{GSegGRAND} yielding an \ac{LLR} pair \((\ell_{\text{ORB},i}, \ell_{\text{GSeg},i})\) for each bit \(i\). Figure~\ref{fig:so-scatter} shows the pairs as a scatter plot with (SO) (left) and (SO+) (right) computation.
\begin{figure}[tbp]
    \centering
    \includegraphics{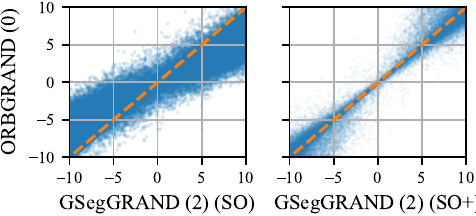}
    \caption{Scatter plot of LLRs with improved SO calculation (SO+) and default SO calculation (SO) for (32, 26) eBCH component code decoding at \(\lEbNO = 3\) dB.}
    \label{fig:so-scatter}
\end{figure}
\ac{GSegGRAND}'s \ac{SO} closely match with \ac{ORBGRAND}'s, i.e., \(\ell_{\text{ORB},i} \approx \ell_{\text{GSeg},i}\), when the scatter points lie on the \(x=y\) diagonal (orange dashed line). There is a significant discrepancy between \ac{GSegGRAND} with the unconditioned \ac{SO} calculation and \ac{ORBGRAND}: when \ac{ORBGRAND} outputs \acp{LLR} with magnitude smaller \(5\), the unconditioned \ac{SO} even frequently produces \acp{LLR} with opposite sign.
On the other side, \ac{GSegGRAND} (SO+) almost always produces \acp{LLR} of the same sign as \ac{ORBGRAND} and matches \ac{ORBGRAND}'s \acp{LLR} closely at low \acp{LLR}.

Figure~\ref{fig:tpc-const-guessing} shows how the (SO) and (SO+) computation affect the turbo product decoding performance for the \((32, 21)^2\) dRM \ac{TPC} and \((32, 26)^2\) eBCH \ac{TPC}.
\begin{figure}[tbp]
    \centering
    \includegraphics{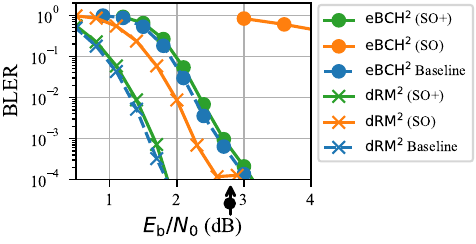}
    \caption{Comparison of TPC decoding with and without improved SO for GSegGRAND (3) component code decoding ((SO+) and (SO), respectively). Legend entry \(\text{eBCH}^2\) refers to the \((32, 26)^2\) eBCH \ac{TPC} and legend entry \(\text{dRM}^2\) to the \((32, 21)^2\) dRM \ac{TPC}.
    With improved SO, TPC decoding performs close to the baseline (ORBGRAND (0) component code decoding without constrained guessing) while the original SO is insufficient to achieve good performance. (Figure~\ref{fig:tpc-const-guessing-var-const} further evaluates the BLER performance for varying number of constraints at the marked \(\lEbNO\) points.)}
    \label{fig:tpc-const-guessing}
\end{figure}
With the proposed \ac{SO} calculation (SO+), turbo product decoding with GSegGRAND (3) performs close to the baseline ORBGRAND (0) (dashed blue). (SO+) gains up to \SI{0.7}{\dB} over (SO) for the \((32, 21)^2\) dRM \ac{TPC} and enables efficient decoding for the \((32, 26)^2\) eBCH \ac{TPC}, which otherwise experiences a high error floor with (SO) and constrained guessing. 
This behavior is consistent for constrained guessing across the number of constraints where Fig.~\ref{fig:tpc-const-guessing-var-const} shows the BLER performance of the \((32, 26)^2\) eBCH \ac{TPC} in the waterfall region (black vertical dashed line in Fig.~\ref{fig:tpc-const-guessing}) for a varying number of constraints. In this plot, 2-4 constraints correspond \ac{GSegGRAND} and one constraint corresponds to \ac{ORBGRAND} (1), which skips noise effects of incorrect parity for even codes (see Sec.~\ref{sec:socalculation}).
Note that not only \ac{GSegGRAND} benefits from (SO+) but also \ac{ORBGRAND} (1) if it skips noise effects of incorrect parity for even codes (see Sec.~\ref{sec:socalculation}), requires (SO+).
\begin{figure}[tbp]
    \centering
    \includegraphics{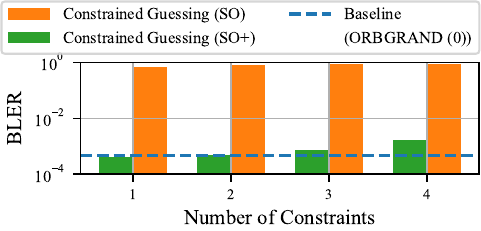}
    \caption{Comparison of decoding for the \((1024, 676)=(32, 26)^2\) eBCH \ac{TPC} with and without improved SO calculation using constrained component code guessing with varying numbers of constraints. BLER is evaluated at a single point in the waterfall region. One constraint corresponds to ORBGRAND (1), while 2-4 constraints correspond to GSeGRAND. Regardless of the number of constraints, (SO) degrades in performance while (SO+) achieves performance close to baseline.}
    \label{fig:tpc-const-guessing-var-const}
    \vspace{-0.2cm}
\end{figure}

\section{Conclusion}
In this paper, we presented \ac{GSegGRAND}, a novel generalization of Segmented \ac{GRAND} that efficiently incorporates multiple parity check constraints during guessing. The underlying parity check matrix structure is designed such that arbitrary linear codes can be systematically transformed into it, facilitating the extraction of multiple constraints. For iterative \acp{TPC} decoding, we demonstrate that the unmodified \ac{SOGRAND} \ac{SO} computation becomes inaccurate under constrained guessing, and derive a novel, accurate \ac{SO} formula for \ac{GSegGRAND} and Segmented GRAND. By explicitly incorporating the constraints in the \ac{SO} computation, \ac{GSegGRAND} achieves guesswork reduction of up to \(88 \%\) under turbo product decoding while maintaining near \ac{ORBGRAND} \ac{BLER} performance.
While we focus here on the practical case of soft-input decoding, the balanced tree transformation provides a general framework for GRAND guesswork reduction that is applicable to various applications, such as hard-decision decoding. Since hard-decision decoding has no discrepancy between \eqref{eq_ORB_rel} and \eqref{eqn:rel-gen-seg-grand}, \ac{GSegGRAND} can reduce guesswork without sacrificing BLER.
Our evaluation demonstrates that \ac{GSegGRAND} is a promising candidate for low-latency and high-performance decoding in future communication systems, offering a systematic path to \ac{GRAND} guesswork reduction across a broad class of linear codes.
\vspace{-0.1cm}


\end{document}